\documentclass[aps,prx,reprint,twocolumn,floatfix]{revtex4-2}
\usepackage{iceberg}
\usepackage{graphicx}
\usepackage[caption=false]{subfig}
\usepackage{subfig}
\usepackage{amsmath}
\usepackage{amssymb}
\usepackage{amsthm}
\usepackage{multirow}
\usepackage[sort&compress]{natbib}

\usepackage{etoolbox}

\makeatletter
\patchcmd{\@makecaption}{\@ifdim{\wd\@tempboxa >\hsize}}{\@firstoftwo}{}{}
\makeatother

\begin{document}

\frenchspacing

\title{The Pinnacle Architecture: Reducing the cost of breaking RSA-2048 to 100 000 physical qubits using quantum LDPC codes}

\author{Paul Webster}
\email{paul@iceberg-quantum.com}
\author{Lucas Berent}
\author{Omprakash Chandra}
\author{Evan T.\ Hockings}
\author{Nou\'edyn Baspin}
\author{Felix Thomsen}
\author{Samuel C. Smith}
\author{Lawrence Z.\ Cohen}
\email{larry@iceberg-quantum.com}
\affiliation{Iceberg Quantum, Sydney}

\begin{abstract}
The realisation of utility-scale quantum computing inextricably depends on the design of practical, low-overhead fault-tolerant architectures.
We introduce the \textit{Pinnacle Architecture}, which uses quantum low-density parity check (QLDPC) codes to allow for universal, fault-tolerant quantum computation with a spacetime overhead significantly smaller than that of any competing architecture.
With this architecture, we show that 2048-bit RSA integers can be factored with fewer than one hundred thousand physical qubits, given a physical error rate of $10^{-3}$, code cycle time of $1$ \textmu s and a reaction time of $10$ \textmu s.
We thereby demonstrate the feasibility of utility-scale quantum computing with an order of magnitude fewer physical qubits than has previously been believed necessary.
\end{abstract}

\maketitle

\section{Introduction}
Quantum computers offer the promise of efficient solutions to currently intractable problems with the potential to allow breakthroughs in areas such as cryptography~\cite{shor_algorithms_1994}, materials science and chemistry~\cite{lloyd_universal_1996}.
However, due to the precision required and the significant levels of noise that afflict all engineered quantum systems, this is only possible if quantum computers are fault tolerant~\cite{gottesman_opportunities_2022}. 
Fault-tolerant quantum architectures are therefore a cornerstone of all efforts to build useful quantum computers.

Sophisticated fault-tolerant architectures have been developed based on the surface code~\cite{fowler_surface_2012, horsman_surface_2012, litinski_game_2019, gidney_how_2021, gidney_how_2025}. 
However, these suffer from a very high overhead, since they require hundreds or thousands of physical qubits to encode a single logical qubit with a low enough failure rate to allow utility-scale computations. 
As a result, utility-scale quantum computers based on such architectures are expected to require at least one million physical qubits~\cite{gidney_how_2025, mohseni_how_2025}. 
Scaling quantum hardware to this size poses formidable challenges~\cite{mohseni_how_2025}.
Developing fault-tolerant quantum architectures with lower overhead is therefore of the utmost importance.

To meet this goal, we introduce the \textit{Pinnacle Architecture}.
This architecture achieves substantial spacetime reductions compared with prior state-of-the-art architectures~\cite{gidney_how_2025,yoder_tour_2025} by reducing the space overhead relative to surface code architectures without a commensurate increase in time overhead.
We realise these savings through the use of processing units constructed from bridged QLDPC code blocks equipped with modular, efficient gadgets for performing gates by generalised surgery~\cite{webster_explicit_2025}.
We also introduce a new component---the \textit{magic engine}---which exploits the multiple logical qubits of QLDPC codes to simultaneously support magic state distillation and injection in a single code block, allowing for constant throughput of high-fidelity magic states with low overhead.
Moreover, we use the technique of Clifford frame cleaning~\cite{chamberland_universal_2022} to develop a method for efficient parallelism of operations across processing units.
In particular, this allows for parallel access to a quantum memory that can allow for spacetime overhead reductions by enabling algorithms to be parallelised by duplicating processing units while keeping only a single memory.
Scalability and hardware-compatibility is also ensured through a modular structure which ensures that connectivity between physical qubits is only required on a length scale constant in the number of logical qubits.
We summarise the main features of the architecture in \cref{sec:summary-architecture} and, after reviewing relevant background concepts in \cref{sec:background},  we then provide a full presentation of the general architecture \cref{sec:architecture}, and a specific instantiation in \cref{sec:instantiation}.

We benchmark the performance of the Pinnacle Architecture by presenting a compilation to a standard application: factoring 2048-bit RSA integers~\cite{gidney_how_2021, gidney_how_2025}.
With standard hardware assumptions (i.e., a physical error rate of $10^{-3}$, code cycle time of 1 \textmu s and reaction time of 10 \textmu s), we show that factoring can be achieved with fewer than one hundred thousand physical qubits.
This outperforms the previous best result by an order of magnitude~\cite{gidney_how_2025}.
We further show the broad applicability of the architecture by showing that it allows for classically intractable instances of the problem of determining the ground-state energy of the Fermi-Hubbard model to be achieved with tens of thousands of physical qubits, under the same assumptions.
Again, this amounts to an order-of-magnitude improvement on the best prior end-to-end resource estimates~\cite{kivlichan_improved_2020}.
In addition to these results, we also apply our compilations under alternative hardware assumptions to provide optimised resource estimates applicable across a range of hardware platforms.
We summarise these results in \cref{sec:summary-applications} and present further details in \cref{sec:applications}.

Through the Pinnacle Architecture, we thus open the door to utility-scale quantum computing on one hundred thousand physical qubit devices.
This has the potential to significantly accelerate the timescale for commercialised and impactful quantum computers.

\section{Summary of Contributions}
\subsection{Pinnacle Architecture}
\label{sec:summary-architecture}
The Pinnacle Architecture consists of:
\begin{itemize}
    \item \textit{Processing Units} consisting of bridged processing blocks each constructed from a QLDPC code block with an ancillary measurement gadget system~\cite{webster_explicit_2025}.
    An arbitrary logical Pauli product measurement can be performed on the logical qubits of the unit in each logical cycle.
    \item \textit{Magic Engines} consisting of a QLDPC code block along with ancillary systems for injecting noisy magic states.
    In each logical cycle, the magic engine stores a high-fidelity magic state as it is injected into a processing unit in parallel with hosting magic state distillation to prepare a high-fidelity magic state for the next logical cycle. 
    It thereby provides a high-fidelity magic state per logical cycle for each processing unit to allow for universal quantum computing.
    \item \textit{Memory} included as an optional component that allows for especially low overhead quantum storage in code blocks, which can be accessed by processing units via ports.
\end{itemize}
\Cref{fig:architecture} shows how the Pinnacle Architecture is assembled from these modules.

Compilation is performed via Pauli-based computation~\cite{bravyi_trading_2016}.
This allows for universal fault-tolerant quantum computing on arbitrarily many logical qubits with a time cost that scales with the $T$ count.

\begin{figure*}[]
    \centering
    \subfloat[Example of the Pinnacle Architecture with one processing unit and approximately one hundred thousand physical qubits.\label{architecture-a}]{ \includegraphics[width=\textwidth]{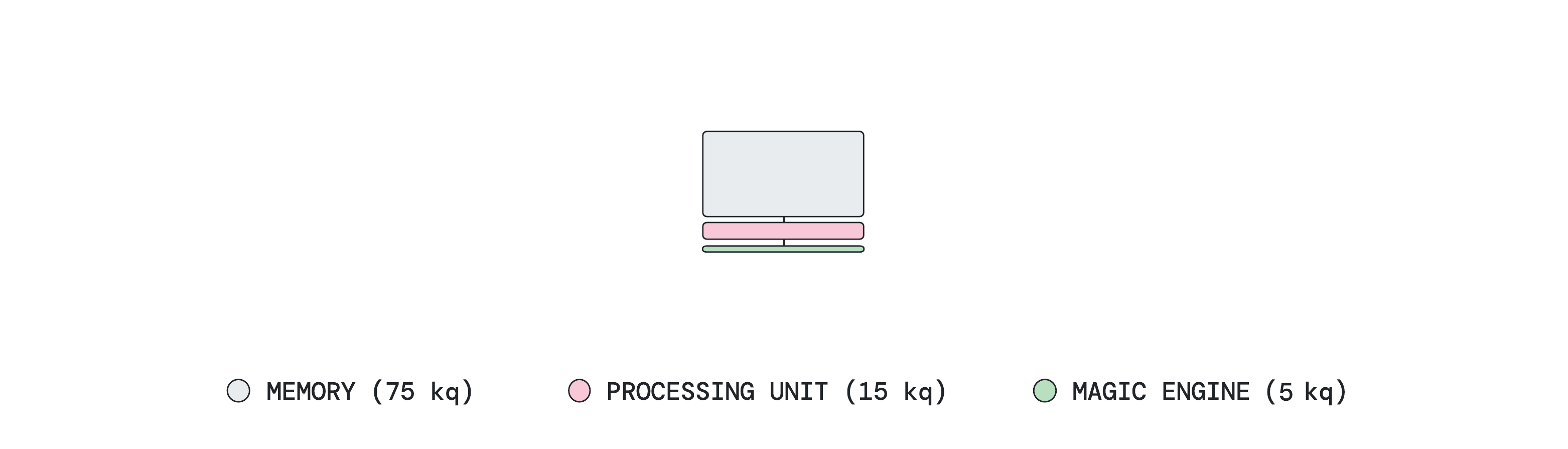} } 
    \vspace{-1em}
    \hfill
    
    \subfloat
    [
        Example of the Pinnacle Architecture with 81 processing units and approximately one million physical qubits.\label{architecture-b}
    ]
    {
        \includegraphics[width=\textwidth]{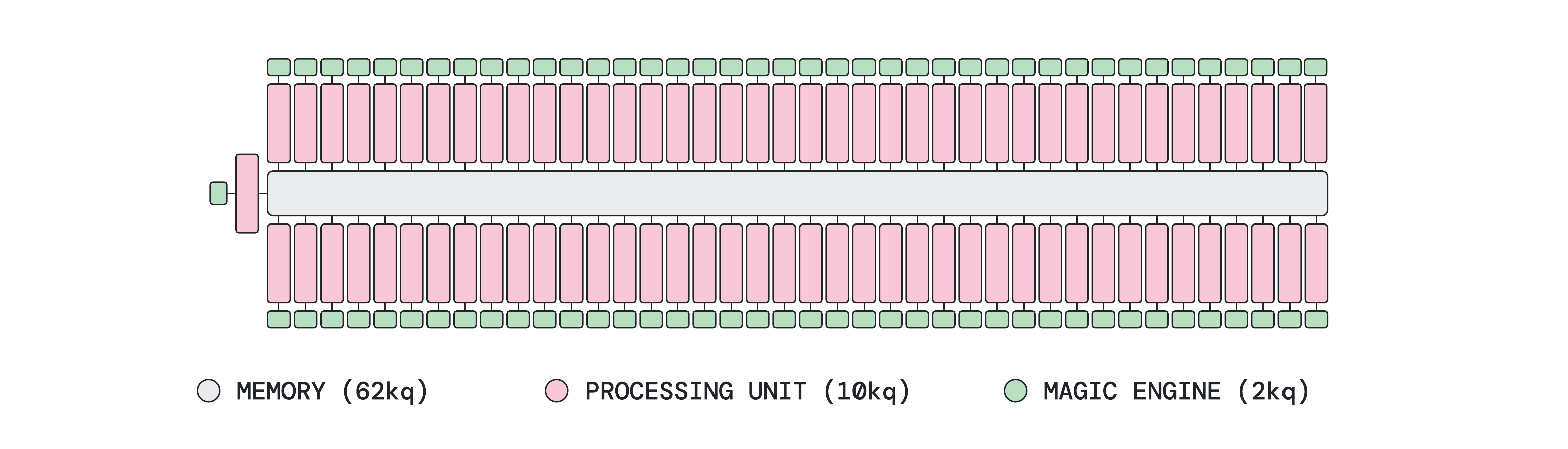}
    } 
    \caption{Examples of the Pinnacle Architecture. These two examples represent specific examples of different space-time trade-offs and code block choices, optimised for RSA-2048 factoring in different hardware regimes.
    Example \textbf{(a)} allows factoring in one month with a physical error rate of $p=10^{-3}$ and a code cycle time of $t_c=1$ \textmu s.
    Example \textbf{(b)} allows factoring in three months with a physical error rate of $p=10^{-4}$ and a code cycle time of $t_c=1$ ms. 
    Shorter runtimes can be achieved by adding more processing units, increasing paralellisation at the cost of additional physical qubits.}
    \label{fig:architecture}
\end{figure*}

Features of the Pinnacle Architecture include:
\begin{itemize}
\item \textit{Low Spacetime Overhead:} The use of QLDPC codes allows order-of-magnitude reductions in physical qubit number compared with surface code architectures. 
This is achieved without the corresponding increase in time overhead of the previous state-of-the-art QLDPC code architecture---the \textit{bicycle architecture} of Ref.~\cite{yoder_tour_2025}--through the use of efficient gadget systems that allow for arbitrary logical Pauli measurements instead of only a subset of them.
It therefore substantially reduces the total spacetime overhead compared to previous fault-tolerant architectures.
\item \textit{Limited Connectivity and Routing:} 
The architecture requires only interactions between physical qubits on the scale of a processing block, which is constant in the number of logical qubits.
This means that it does not depend on all-to-all connectivity but instead is implementable on hardware platforms that support quasi-local connections between physical qubits separated by a bounded distance.
Moreover, arbitrary quantum circuits can be performed efficiently using a static configuration of processing blocks without depending on long-distance routing.

\item \textit{Modularity and Parallelism:} Computations can be separated across multiple processing units with limited connectivity between them, allowing compatibility with modular hardware.
Efficient parallelism of non-Clifford gates is also supported to supplement this modular structure.
In particular this also allows parallel, read-only access to memory by multiple processing units.
\end{itemize}

The Pinnacle architecture therefore offers a new alternative architecture that is both practical for implementation on a range of hardware platforms and significantly more efficient than previous state-of-the-art alternatives.

\subsection{Results}
\label{sec:summary-applications}
For concreteness, we present a specific instantiation of the Pinnacle Architecture using the family of generalised bicycle codes and the modular, efficient measurement gadgets introduced in Ref.~\cite{webster_explicit_2025}.
Based on compilation to this instantiation, we present resource estimates for two benchmark applications.

\begin{figure}[]
    \centering
    \includegraphics[width=\linewidth]{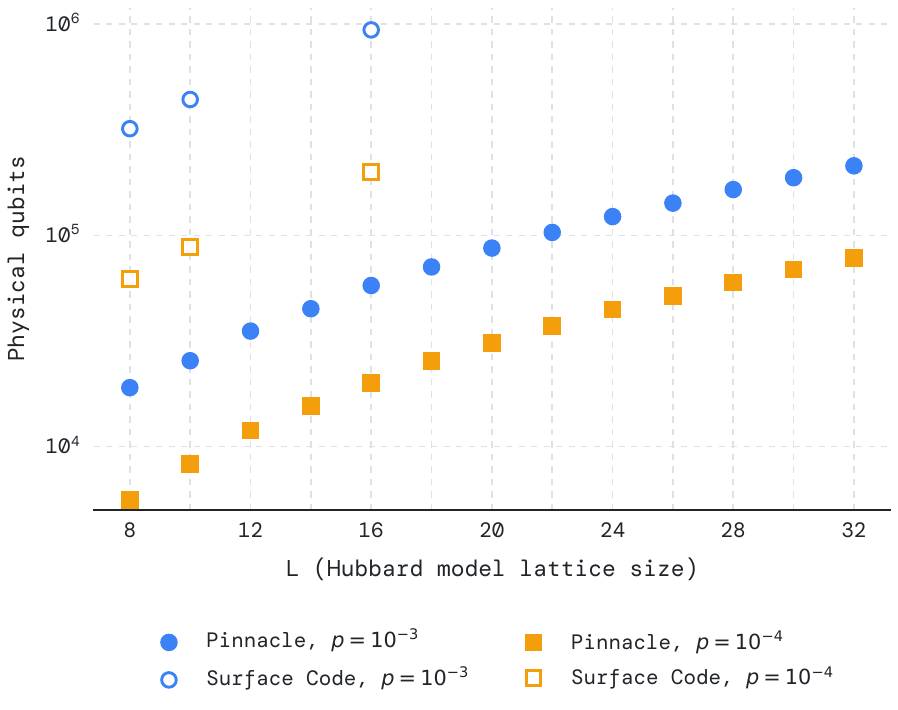}
    \caption{Physical qubits required for determining the ground state energy of the Fermi-Hubbard model on an $L\times L$ lattice to 0.5\% relative precision.
    Surface code values correspond to the minimum quoted number of physical qubits with $u/\tau=4$ in Ref.~\cite{kivlichan_improved_2020}.
    \label{fig:FH-graph}}
\end{figure}

\begin{figure*}
    \centering
    \subfloat[\label{fig:heatmaps:1e-3}]{
        \includegraphics[width=0.9\textwidth]{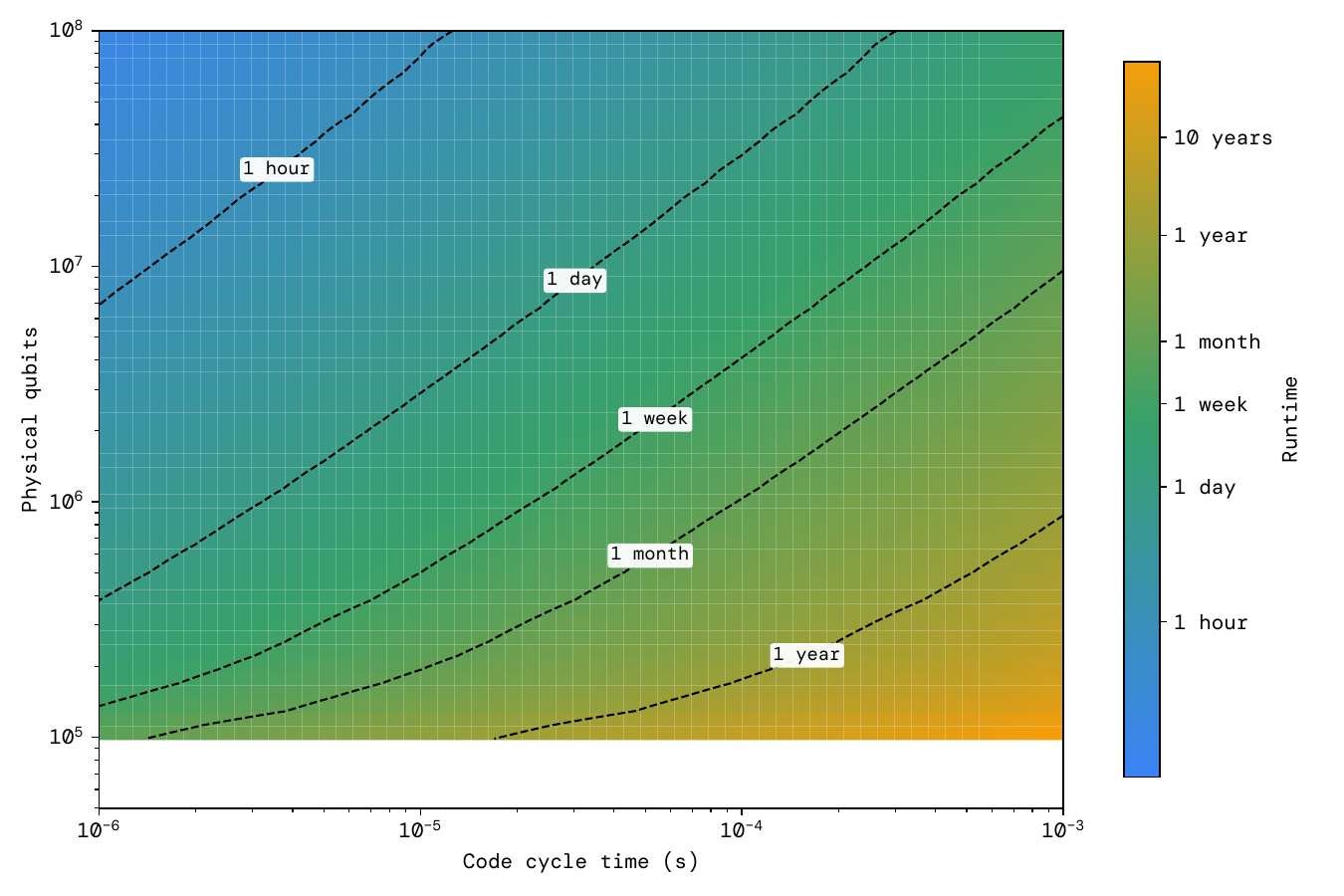}
    }
    \hfill
    \subfloat[\label{fig:heatmaps:1e-4}]{
        \includegraphics[width=0.9\textwidth]{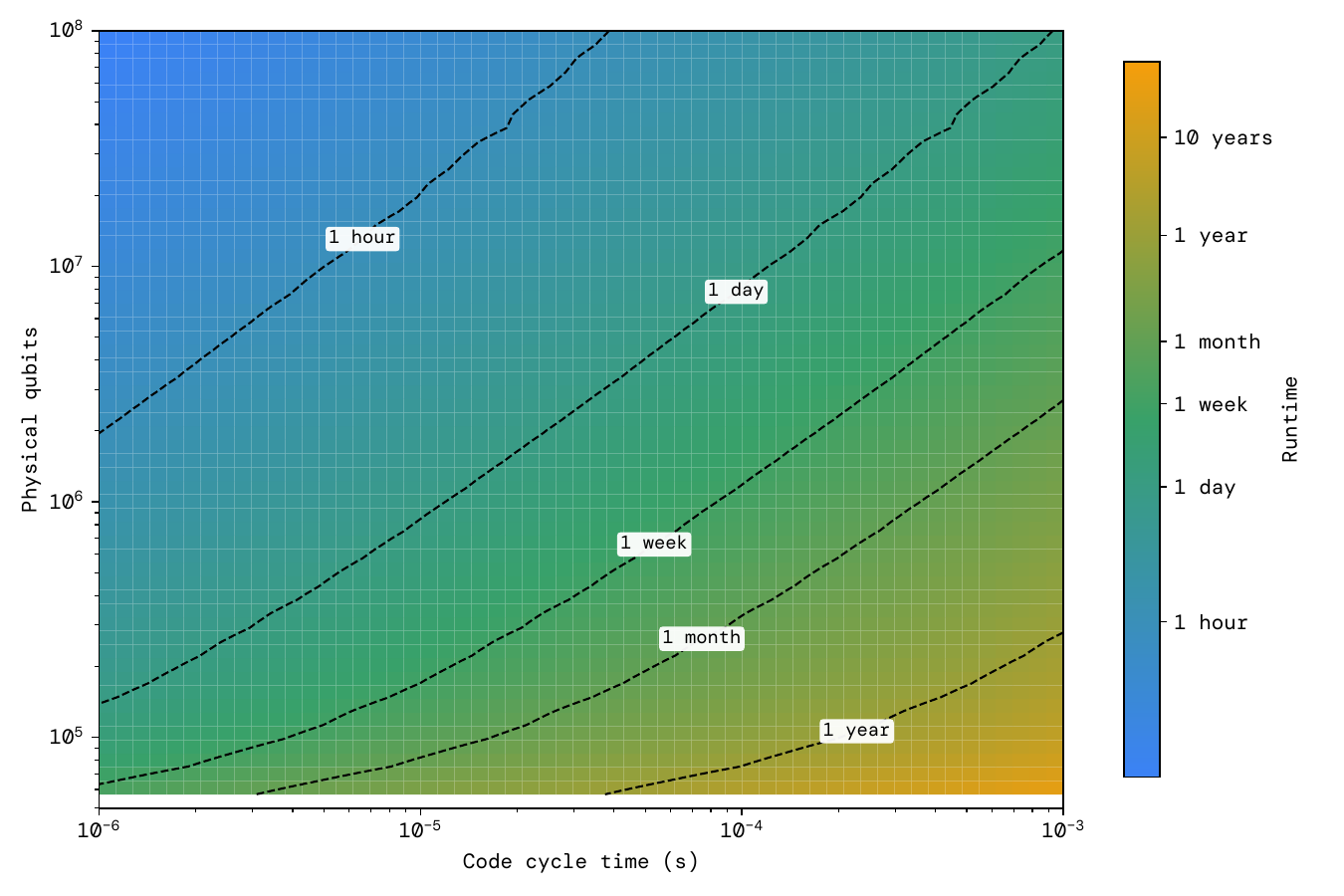}
    }
    \caption{Optimal expected runtime for factoring an RSA-2048 integer on the Pinnacle Architecture as a function of the number of physical qubits and the code cycle time at physical error rates of \textbf{(a)} $p=10^{-3}$ and \textbf{(b)} $p=10^{-4}$.
    White areas indicate insufficient physical qubits to implement the algorithm. 
    The reaction time in all cases is assumed to be equal to ten times the code cycle time.}
    \label{fig:heatmaps}
\end{figure*}

First, we consider determination of the ground state energy of the Fermi-Hubbard model via plaquette Trotterisation~\cite{campbell_early_2022}.
As shown in \cref{fig:FH-graph}, we achieve order-of-magnitude reductions in physical qubit number relative to the best available end-to-end surface code analysis~\cite{kivlichan_improved_2020}.
For example, we find that at a lattice size of $L=16$ (with a coupling strength of $u/\tau=4$), only 58 thousand physical qubits are required at a physical error rate of $p=10^{-3}$ and only 20 thousand at $p=10^{-4}$.
This compares with 940 thousand and 200 thousand, respectively, in Ref.~\cite{kivlichan_improved_2020}.
We achieve these results while maintaining a modest runtime per shot of 1--4 minutes with microsecond code cycle times or 1--3 days with millisecond code cycle times.

Second, we analyse factoring RSA integers using an algorithm based on that of Ref.~\cite{gidney_how_2025}.
\cref{fig:heatmaps} shows the required physical qubits and runtime to factor a 2048-bit integer for different code cycle times and physical error rates.
With standard assumptions of a physical error rate of $p=10^{-3}$ and a code cycle time of $1$ \textmu s and a $10$ \textmu s reaction time (see \cref{sec:background} for definitions of these timescales), factoring is possible with fewer than one hundred thousand physical qubits, compared with the previous best result of close to one million physical qubits~\cite{gidney_how_2025}.
Moreover, through parallelising the algorithm, we achieve an efficient spacetime trade-off that allows for low overhead factoring in feasible runtimes even with longer code cycle times.
For example, with a code cycle time of 1 ms, factoring can be completed in one month with 2.7 million physical qubits at a physical error rate of $p=10^{-4}$ (typical of trapped ions~\cite{hughes_trappedion_2025}) or with 9.5 million physical qubits at a physical error rate of $p=10^{-3}$ (typical of, for example, neutral atoms~\cite{zhou_resource_2025}).

We therefore conclude that the Pinnacle Architecture can be used to achieve utility-scale quantum computation with significantly reduced overhead across multiple applications and a range of hardware regimes.

\section{Background}
\label{sec:background}
In this section we review relevant concepts of fault-tolerant quantum computation with QLDPC codes.
\subsection{Code Blocks}
A code block is an instantiation of an $\llbracket n,k,d\rrbracket$ quantum error-correcting code. 
Error correction is facilitated by repeatedly performing a syndrome extraction circuit on the code block.
This circuit involves measuring a set of parity check operators, which collectively yield an error syndrome.
This is done with the use of $n_c$ ancilla qubits; in each code cycle, each of these is entangled with the code qubits in accordance with one of the parity check operators and then destructively measured.
A code block therefore requires a total of $n_{cb}=n+n_c$ physical qubits.
In order to ensure robustness against measurement errors, the results of $d_t=\Theta(d)$ code cycles must be combined to yield a reliable error syndrome; this is referred to as a \textit{logical cycle}.

We assume the use of QLDPC codes~\cite{mackay_sparse_2004, breuckmann_quantum_2021}.
These are defined by having low parity check operator weights and qubit degrees (i.e., the number of parity checks supported on each qubit).
Precisely, these check weights and qubit degrees are bounded by a constant independent of the code distance.
This ensures that syndrome extraction circuits can have constant depth, and therefore be fault tolerant.
The most widely used QLDPC codes are distance-$d$ (rotated) surface codes~\cite{fowler_surface_2012}, which are $\llbracket d^2,1,d \rrbracket$ codes with $n_c=d^2-1$ parity checks which have weight at most four and require only nearest-neighbour interactions. 
Surface code blocks therefore use $2d^2-1$ physical qubits to encode one logical qubit.
However, by relaxing the requirement of nearest-neighbour interactions, more general QLDPC codes can allow for many logical qubits to be encoded in a single code block.
This can allow for significant reductions in the overhead of physical qubits required per logical qubit~\cite{breuckmann_quantum_2021}.
The required non-local interactions are supported on a range of hardware platforms~\cite{yoneda_coherent_2021, malinowski_how_2023, bluvstein_quantum_2022, bombin_interleaving_2021, bravyi_highthreshold_2024}.

\subsection{Processing Blocks}
To allow for fault-tolerant quantum computation, instead of only passive storage of quantum information, the concept of a code block must be generalised to a \textit{processing block}.
A processing block allows for logical operations to be implemented on its encoded logical qubits. 

A QLDPC code block can be turned into a processing block by appending a measurement gadget system that allows for \textit{generalised lattice surgery}~\cite{cohen_lowoverhead_2022, williamson_lowoverhead_2024, ide_faulttolerant_2025}.
This construction ensures that the combined code block-gadget system constituting the processing block remains a QLDPC code, while also ensuring that measuring a selected Pauli logical operator of the code is equivalent to the product of a set of parity check measurements on the gadget system.
This allows for a logical Pauli to be measured in parallel with error correction within a logical cycle by performing a modified circuit for syndrome extraction on the full processing block
\footnote{Following Refs~\cite{yoder_tour_2025,litinski_game_2019}, we assume that the time required between logical measurements is negligible compared with the timescale of a logical cycle.}.
Arbitrary logical Pauli measurements can then be performed across multiple processing blocks by bridging the gadget systems of these blocks~\cite{cross_improved_2025, swaroop_universal_2026}. 
The number of physical qubits in the processing block is given by $n_{pb}=n_{cb}+n_G+n_b$, where $n_G$ is the number of physical qubits in the gadget system and $n_b$ is the number of physical qubits used to bridge it to another processing block.

This processing block construction is underpinned by a rich history of prior work on generalised lattice surgery.
For brevity, we have here outlined only the most relevant aspects, but we refer the reader to Ref.~\cite{he_extractors_2025} (particularly Sec.~3.2) for a more complete review of this literature.

\subsection{Pauli-Based Computation}
\label{sec:PBC}
When combined with injected magic states, logical measurements across bridged processing blocks support universal quantum computation on the encoded logical qubits using Pauli-based computation~\cite{bravyi_trading_2016}.
Indeed, a quantum circuit on $\kappa$ qubits with a $T$ count of $\tau$ and any number of Clifford gates can be performed using $\tau+\kappa$ Pauli measurements and one $\ket{T}$ state for each of the first $\tau$ measurements~\cite{litinski_game_2019}.
More generally, if the circuit also contains $o$ intermediate Pauli measurements on which later operations adaptively depend, then the circuit can be performed using $\tau+\kappa+o$ Pauli measurements without any additional $\ket{T}$ states.
This implies that such a circuit can be performed fault-tolerantly using $\lceil \kappa/k \rceil$ $\llbracket n,k,d\rrbracket$ bridged processing blocks in $\tau +\kappa+o$ logical cycles.

The compilation which allows for this implementation follows that presented in Ref.~\cite{litinski_game_2019}.
First, each $T$ gate is replaced by a magic state injection circuit, which requires one Pauli measurement with support on the processing block.
Then, all Clifford gates are commuted through to the end of the circuit and absorbed into the final measurement of each of the $\kappa$ qubits.
This is done using the rule that the Pauli defining the basis of each measurement is transformed by conjugation by each Clifford that either passes through it or (in the case of the final measurements) that it absorbs.
By definition, Clifford gates map Pauli operators to Pauli operators under conjugation, so the resulting circuit consists of $\tau+\kappa+o$ Pauli measurements, along with a $\ket{T}$ state for each $T$ gate, as required.

\subsection{Relevant Timescales}
There are three relevant timescales that we use in the determination of the runtime of a fault-tolerant quantum circuit.
First, there is the code cycle time $t_c$, which is the time required for the completion of one code cycle (i.e., one round of syndrome extraction).
This time is hardware-dependent, with typical estimates for different platforms ranging from 1 \textmu s to 1 ms~\cite{acharya_quantum_2025,bombin_interleaving_2021,zhou_resource_2025,pataki_compiling_2025,leone_resource_2025}.
Second, there is the logical cycle time $t_l$, which is the time required for one logical cycle.
Since a logical cycle consists of $d_t$ code cycles, this is related to the code cycle time by $t_l=d_tt_c$.

Finally, there is the \textit{reaction time} $t_r$~\cite{gidney_how_2021}.
This is defined as the minimum time between the start of one logical measurement, $M$ and the start of any subsequent measurement whose basis depends adaptively on the outcome of $M$.
The reaction time is dependent on the classical control system of the quantum hardware.
For simplicity, we assume throughout that the reaction time is equal to ten times the code cycle time, i.e., $t_r=10t_c$.
Since \mbox{$t_c\geq 1$ \textmu s} for all resource estimates we present, this implies that the reaction time always exceeds the conventionally-assumed minimum value of $10$ \textmu s~\cite{gidney_how_2021}.
With this reaction time, our architecture is not reaction-limited provided that $d_t\geq 10$ for all code blocks which is true for all resource estimates we present, with one exception addressed that is noted and addressed in \cref{sec:gb-me}.

\section{The Pinnacle Architecture}
\label{sec:architecture}
In this section we present the \textit{Pinnacle Architecture}, a low-overhead, modular and parallelisable quantum computing architecture based on QLDPC codes.
We begin by presenting the modules that constitute the architecture---processing units, magic engines and memory---and then describe the overall operation, structure, and scalability of the architecture.

\subsection{Modules}
\subsubsection{Processing Units}
The primary modules of the architecture are \textit{processing units}.
A processing unit uses $\beta$ processing blocks of an $\llbracket n,k,d\rrbracket$ QLDPC code to allow fault-tolerant quantum computation on $\kappa:=\beta k$ logical qubits.
These processing blocks can be arranged in a line with bridges connecting nearest-neighbour blocks.
This allows for the measurement of an arbitrary logical Pauli operator, supported on any or all of the logical qubits in the processing unit, in each logical cycle.

\subsubsection{Magic Engines}
\label{sec:me}
To allow for universal computation, we introduce a new module, which we refer to as a \textit{magic engine}.
A magic engine allows magic states to simultaneously be produced and consumed to provide a continuous throughput of magic states to an associated processing unit.

Specifically, each processing unit is equipped with one magic engine. 
The magic engine produces one encoded $\ket{\bar{T}}$ magic state in each logical cycle, and is also bridged to its associated processing unit to allow a joint measurement with that unit in parallel with state production. 
This is intended to ensure that in each logical cycle there is a magic state available for the processing unit to consume (i.e., the state produced in the previous logical cycle). 

A magic engine can be constructed from an $\llbracket n_e,k_e,d_e \rrbracket$ QLDPC code block as follows.
Partition the logical qubits of the code block into two halves, which we label the left ($L$) and right ($R$) logical sectors.
In odd-numbered logical cycles, a magic state is produced by performing a magic state distillation circuit consisting of a sequence of magic state injections.
These are implemented by joint logical measurements on the logical qubits of the $L$ logical sector and small ancilla systems that hold noisy $\ket{T}$ states.
The result is that the first logical qubit of the $L$ logical sector is in an encoded $\ket{\bar{T}}$ state at the end of the logical cycle.
In parallel, a magic state is consumed from logical sector R by performing a logical measurement of an arbitrary logical operator on the processing unit joint with $\bar{Z}$ on the first logical qubit of $R$.
This allows for the injection of the $\ket{\bar{T}}$ state to perform arbitrary $\pi/8$ Pauli rotations on the processing unit.
To complete the injection, an $\bar{X}$ measurement is also required on the first logical qubit of $R$; a change of basis can be performed between logical cycles to allow that logical qubit to be offline while it is performed in parallel with the continued operation of the engine.
In even-numbered logical cycles, the roles of the two logical sectors are swapped. 

Since magic state distillation relies on post-selection, for each logical cycle there is some probability $p_r$ that the state produced by the magic engine is rejected.
When this happens---and if a magic state is required for the next logical cycle---the processing unit can be left idle for the next logical cycle to allow for a new magic state to be prepared.
This causes the expected number of logical cycles required per $T$ gate to increase from $1$ to $\alpha=(1-p_r)^{-1}$.
The distillation protocol should be chosen such that $p_r$ is small to ensure that this effect is also small.

\begin{figure}[]
    \centering
    \includegraphics[width=\linewidth]{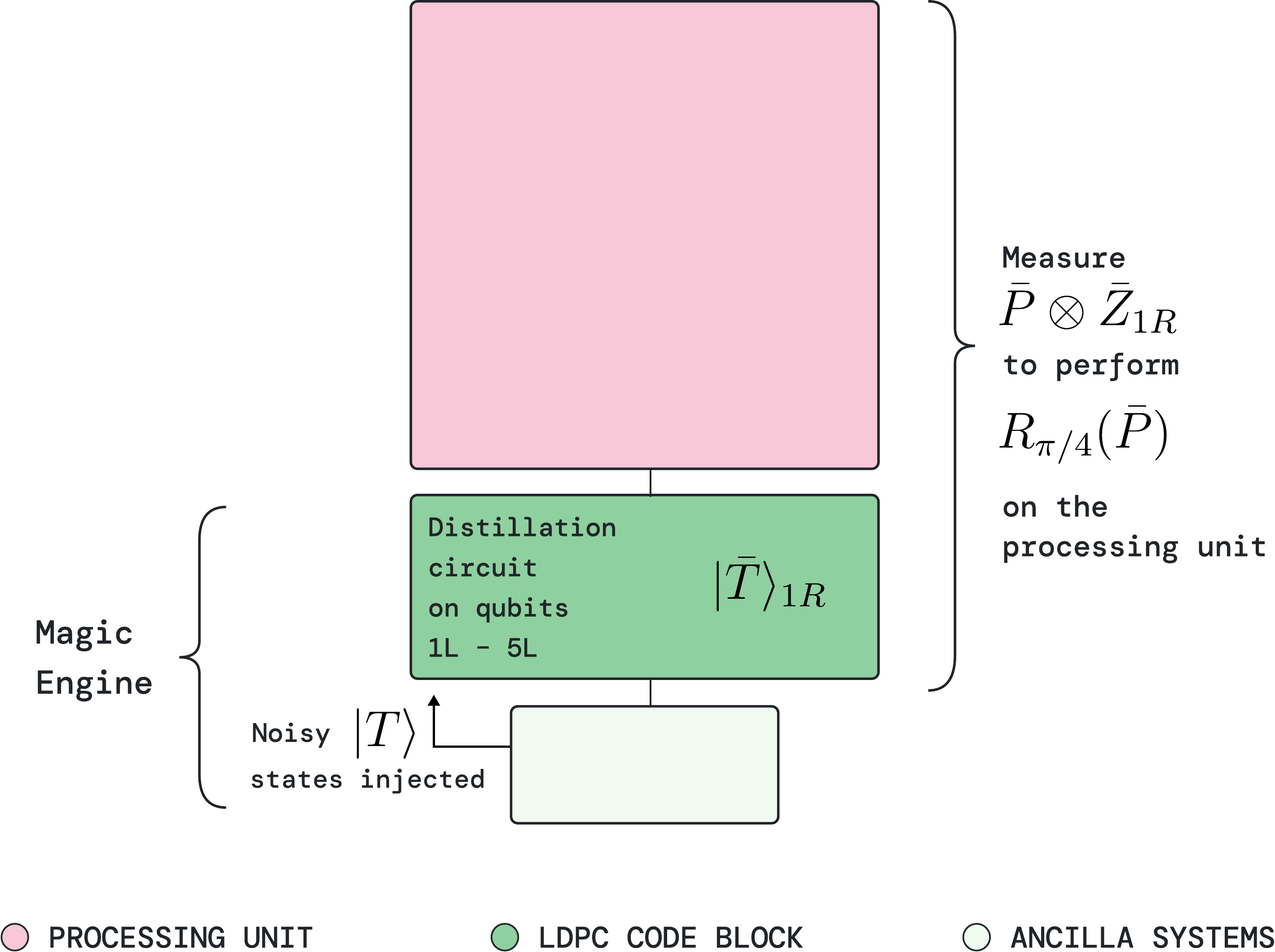}
    \caption{Structure and operation of a magic engine.
    Magic state distillation is applied on one logical sector of a QLDPC code block using noisy $\ket{T}$ states injected from ancillary systems.
    In parallel, an arbitrary Pauli measurement on the processing unit joint with $\bar{Z}_1$ on the other logical sector injects an encoded $\ket{\bar{T}}$ state that was distilled in the previous logical cycle.
    \label{fig:magic-engine}}
\end{figure}

\subsubsection{Memory}
\label{sec:memory}
The architecture can also include memory.
This is optional, but it is useful in cases where a large number of logical qubits must be stored but not processed.
It consists of $\nu$ code blocks of an $\llbracket n_m,k_m,d_m\rrbracket$ quantum error-correcting code encoding $\mu=\nu k_m$ logical qubits. 
Since these logical qubits are not processed, full processing blocks are not required.
However, to facilitate rearrangement of the memory, we ensure that each code block has $n_m$ ancilla qubits (including those used for syndrome extraction).
This means that the memory consists of a total of $2\nu n_m$ physical qubits.

Memory is accessed by processing units via \textit{ports}.
To facilitate this, we partition the logical qubits of the memory into sets of size $w$, which we refer to as \textit{windows}.
For simplicity, we enforce the condition that each window is contained within a single code block; this implies that $w$ divides $k_m$ such that there are a total of $\nu k_m/w$ windows.
Each port is associated with one window and consists of a gadget that allows for arbitrary logical $Z$-type measurements on the $w$ logical qubits of that window.
Bridging a processing unit to a port then allows for any circuit with arbitrary gates on the processing unit and controls on the $w$ logical qubits of the memory window to be implemented (via Pauli-based computation).
This is sufficient to allow read-only access to that window of memory by the processing unit~\cite{babbush_encoding_2018}.
More generally, we can assign a port to any subset of the windows of the memory.
Provided the gadgets constituting each of these ports allow for arbitrary logical $Z$ measurements to be performed in parallel, this can allow for up to $\nu k_m/w$ processing units to access the memory in parallel.

For each processing unit to access the full memory, the memory code blocks must be permuted.
To ensure that this does not require arbitrary routing, we impose the constraint that each code block is to be shifted by at most one position per logical cycle.
We may then perform the required permutation without requiring connectivity on a longer scale than the size of a code block as follows.
The $\nu$ memory code blocks are arranged such that the $i$th and $(i+1)$th (mod $\nu$) code blocks are adjacent for $1\leq i\leq \nu$ (e.g., in a loop).
A cyclic shift of memory code blocks is performed by physical SWAPs of the $j$th data qubit in code block $i$ with the $j$th ancilla qubit in code block $i+1$ (mod $\nu$) for $1\leq i\leq \nu$ and $1\leq j\leq n$, followed by a (local) SWAP of the $j$th data qubit and $j$th ancilla qubit in each code block.
Since such a circuit requires only one layer of non-local gates confined to the scale of a code block, whereas a QLDPC syndrome extraction circuit generally requires many such gates, we assume that this rearrangement can be completed during a code cycle and so its time cost is negligible.
By applying $\nu$ of these cyclic shifts over a period of at least $\nu$ logical cycles, every window of memory can be accessed by each processing unit.

\subsection{Operation}
We now consider how the architecture operates while performing a quantum computation.
To aid understanding, we start with a simplified baseline operation (\cref{sec:baseline-operation}) and incrementally build up to the most general operation (\cref{sec:general-operation}).

\subsubsection{Serial Operation}
\label{sec:baseline-operation}
As a baseline, let us first consider a serial mode of operation.
In this mode, there is a single processing unit with $\kappa$ logical qubits (and, for simplicity, we assume there is no memory).
During each logical cycle, a joint logical Pauli measurement is performed on the processing unit and magic engine.
In parallel, the magic engine produces a magic state for the next logical cycle.
Accounting for a magic engine reject rate of $p_r$, this allows for an arbitrary Clifford+T circuit on $\kappa$ qubits with a $T$ count of $\tau$ and $o$ intermediate measurements to be performed fault-tolerantly in an average of $\tau/(1-p_r)+\kappa+o$ logical cycles.

\subsubsection{Fully Parallel Operation}
As a next step, we can consider the case of implementing a circuit that can be completely separated out into two or more independent circuits.
In this context, we can separate the architecture up into a separate processing unit for each independent circuit and perform all the circuits in parallel.
This reduces the number of logical timesteps required from $\tau/(1-p_r)+\kappa+o$ to approximately $\max_{i}{\left(\tau_i/(1-p_r)+\kappa_i+o_i\right)}$ where $\tau_i$, $\kappa_i$ and $o_i$ denote the number of $T$ gates, logical qubits and intermediate logical measurements in the $i$th independent circuit. 
This expression omits an $O\left(\sqrt{\max_{i}{(\tau_i)}}\right)$ correction arising from variance in the proportion of magic states rejected in different processing units across the duration of the circuit.
Since we are interested in circuits where the $T$ count is large, we assume this correction is negligible.

An example of where this mode could be used is in implementing multiple shots of an algorithm in parallel.
In this context, it can be considered a way to use a greater number of qubits to reduce the runtime compared to when all shots are performed in series.

\subsubsection{Flexibly Parallel Operation}
\label{sec:parallel-operation}
More common and general is the case where a circuit can be implemented partially in parallel.
In such a circuit, no subset of logical qubits is entirely separable from the rest, but significant parts of the circuit involve operations on disjoint registers of logical qubits.
A conventional circuit implementation would allow such parts of the circuit to be performed in parallel on the disjoint registers.
We now show how this can be done in the Pinnacle Architecture.

Parallelism of this kind is inherently challenging with Pauli-based computation because of the effect of commuting Clifford gates through to the end.
To see this, consider the case of a \textit{Clifford frame} at a given point in the circuit (i.e., the product of Clifford gates up to that point) that corresponds to an entangling gate between two processing units.
Then, as Cliffords are commuted through the circuit the supports of the logical measurements on one processing unit will spread out to straddle both.
In particular, this means that a logical measurement corresponding to a $\ket{\bar{T}}$ state injection---required to perform a $\bar{T}$ gate on either processing unit---comes to have support on both units.
Since each processing unit only allows one logical measurement per logical cycle, this implies that a $\bar{T}$ gate on a logical qubit on one processing unit cannot be performed in parallel with a $\bar{T}$ gate acting on a logical qubit of the other processing unit.

Parallelism therefore requires that the Clifford frame acts as a tensor product across the units which are to be parallelised.
One way to achieve this (following Ref.~\cite{he_extractors_2025}) could be to perform $\overline{\text{CNOT}}$ gates that entangle processing units physically using additional logical measurements, so that they can be excluded from the Clifford frame.
However, this approach leads to a time cost that scales with the number of $\overline{\text{CNOT}}$ gates, which can quickly cause the benefits of parallelism to be erased.
In particular, it performs poorly in the common setting where one part of the circuit is highly parallelisable but another is not, since the cost of parallelising the former part scales with the number of entangling gates in the latter part.

We instead propose a more flexible alternative that allows for parallelism when it is beneficial but avoids the cost of physically implementing every $\overline{\text{CNOT}}$ gate.
To develop this approach, we use the technique of \textit{Clifford frame cleaning}.
Precisely, let $\mathbf{K}$ be a set of logical qubits and $\mathbf{K}'\subset \mathbf{K}$ be a subset of $\left| \mathbf{K}'\right|$ of these logical qubits.
If $C$ is a Clifford frame acting on $\mathbf{K}$, cleaning $C$ off $\mathbf{K}'$ means physically performing a Clifford $U$ such that $CU$ acts trivially on $\mathbf{K}'$.
We show in \cref{lem:cleaning1} that this can be done using at most $4\left|\mathbf{K}'\right|$ logical Pauli product measurements on $\mathbf{K}$. 
We note that an instance of Clifford frame cleaning was previously introduced in Ref.~\cite{chamberland_universal_2022} for the specific purpose of caching in surface code architectures; the construction we present is more generally applicable to arbitrary generalised surgery architectures.

Harnessing this tool, our flexible parallelism framework is as follows.
We consider processing units to automatically be \textit{joined} into larger units at any point in the circuit where there is an entangling gate between the units.
From that point on, $\bar{T}$ gates and logical measurements on any of the constituent processing units of this joined unit are assumed to require joint logical Pauli product measurements across the entire joined unit, meaning that only one such gate can be implemented on the joined unit per logical timestep.
At any later time, we can then \textit{separate} a processing unit (with $\kappa$ logical qubits) from a joined unit by cleaning the Clifford frame off the joined unit, at a cost of at most $4\kappa$ additional logical timesteps.
From then on, logical measurements on the separated unit can again be performed in parallel with the unit it was separated from.
This process is shown in \cref{fig:cleaning}.

\begin{figure}
\centering
\subfloat[
Two separate processing units initially operate in parallel.
\label{cleaning-a}
]{
    \includegraphics[scale=0.25]{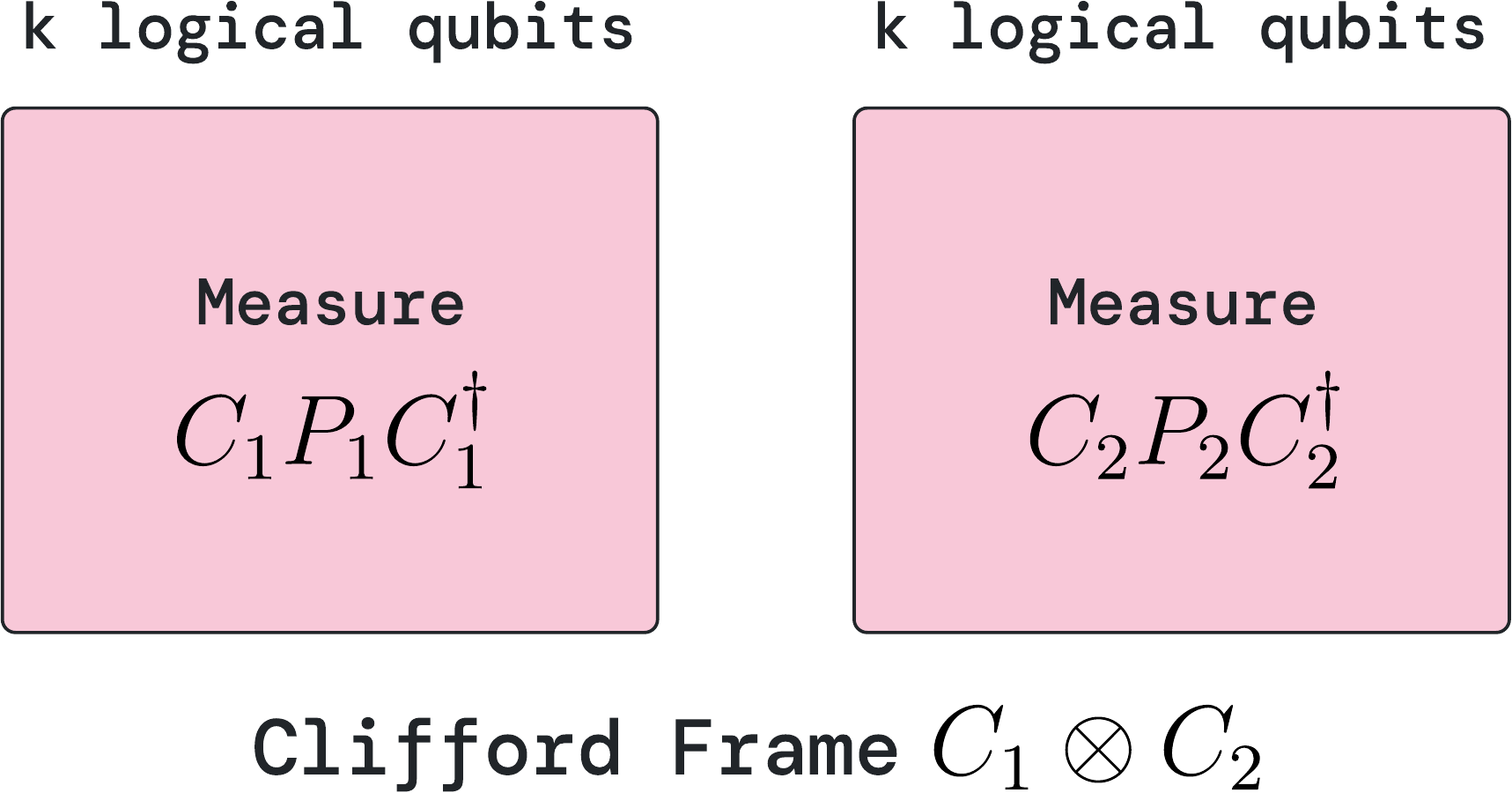}
}
\vspace{0.3cm}
\subfloat[
A point in the circuit is reached where a $\overline{\text{CNOT}}$ straddles the two processing units, causing the Clifford frame to become entangling across the units. 
This \textit{joins} the units, requiring all subsequent logical measurements on either unit to be performed serially.
\label{cleaning-b}
]{
    \includegraphics[scale=0.25]{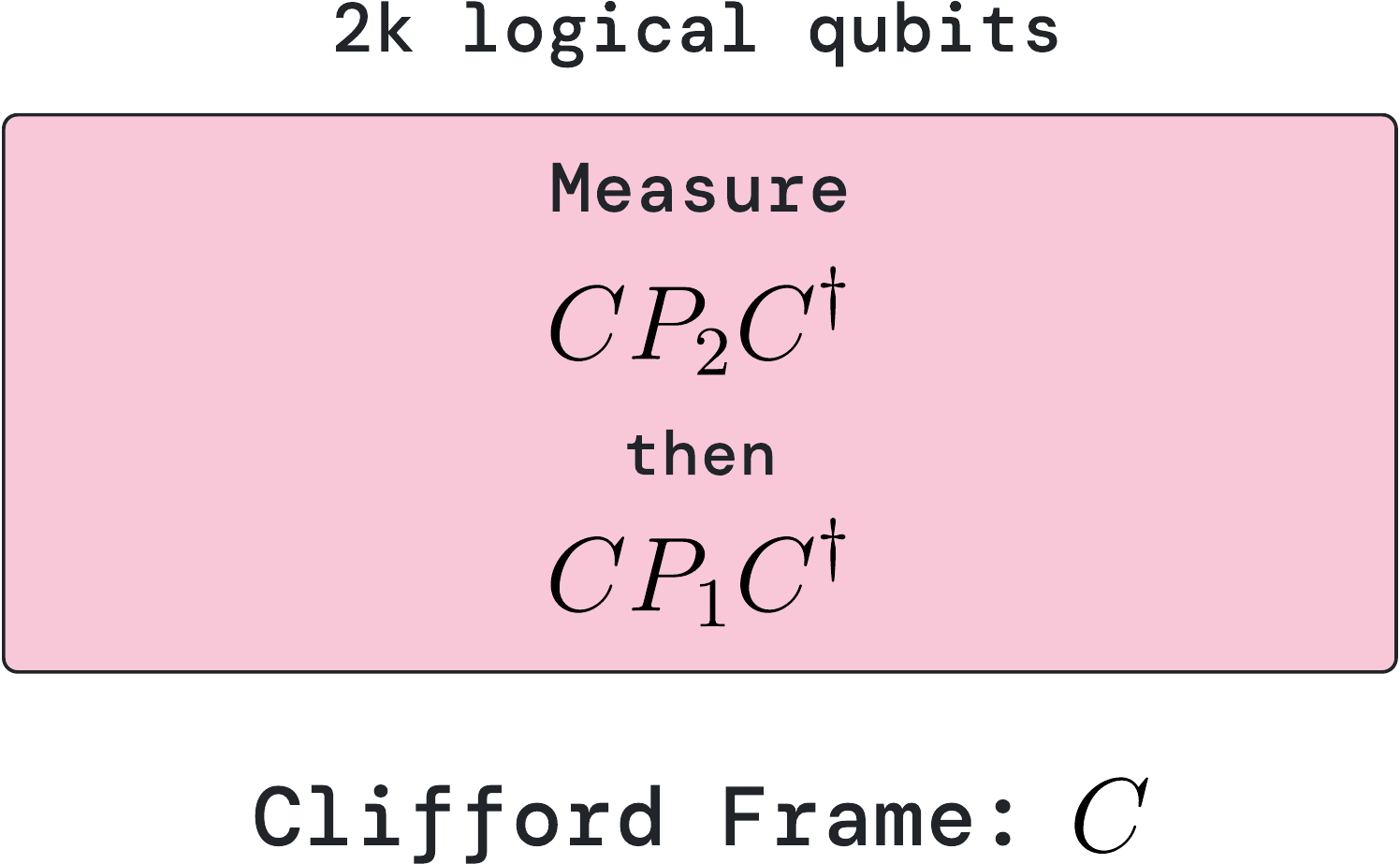}
}
\vspace{0.35cm}
\subfloat[
At any later point, the Clifford frame can be \textit{cleaned} by performing up to $4k$ additional logical measurements on the joined processing units.
\label{cleaning-c}
]{
    \includegraphics[scale=0.25]{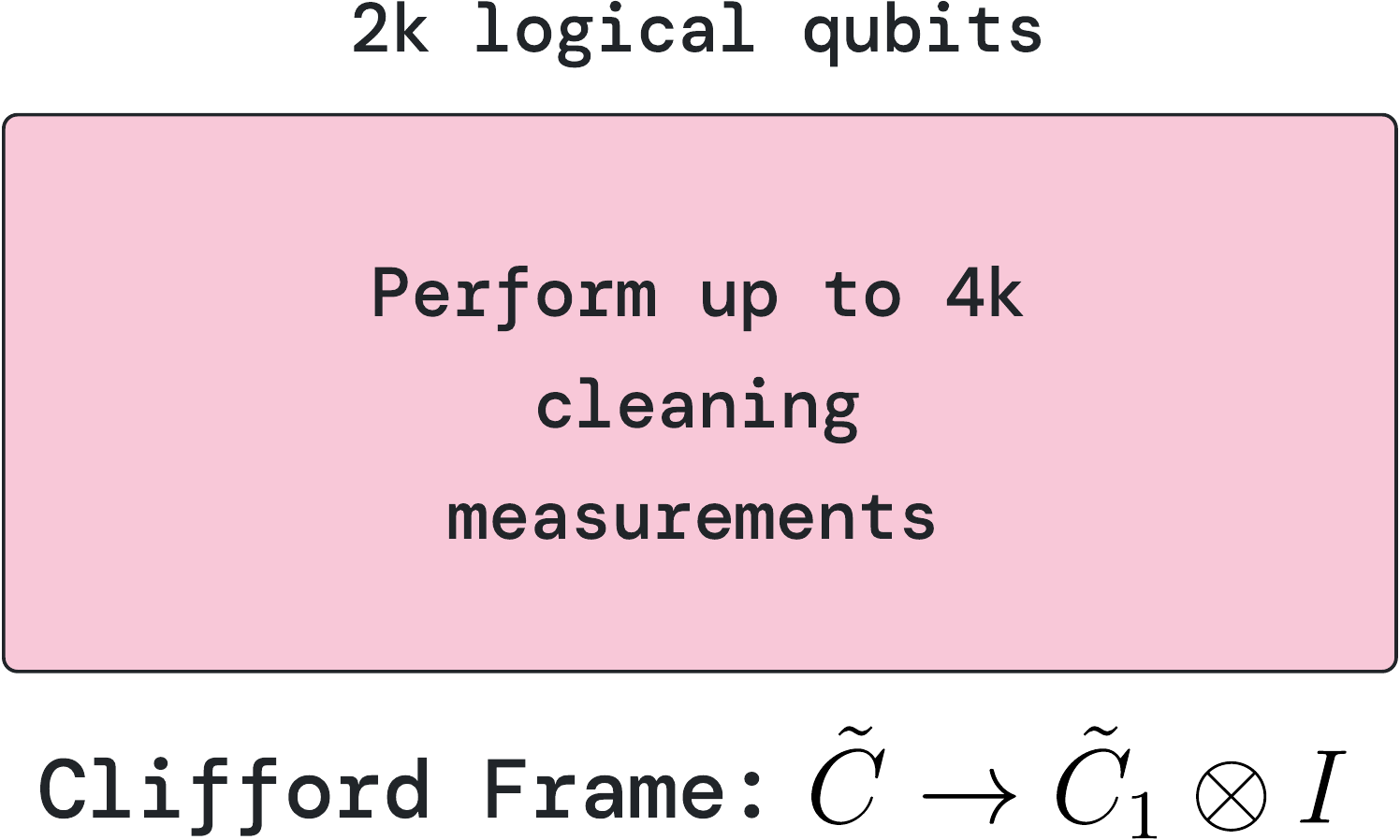}
}
\vspace{0.35cm}
\subfloat[
After Clifford frame cleaning is completed, the two processing units are separated again, allowing parallel measurements to resume.
\label{cleaning-d}
]{
    \includegraphics[scale=0.25]{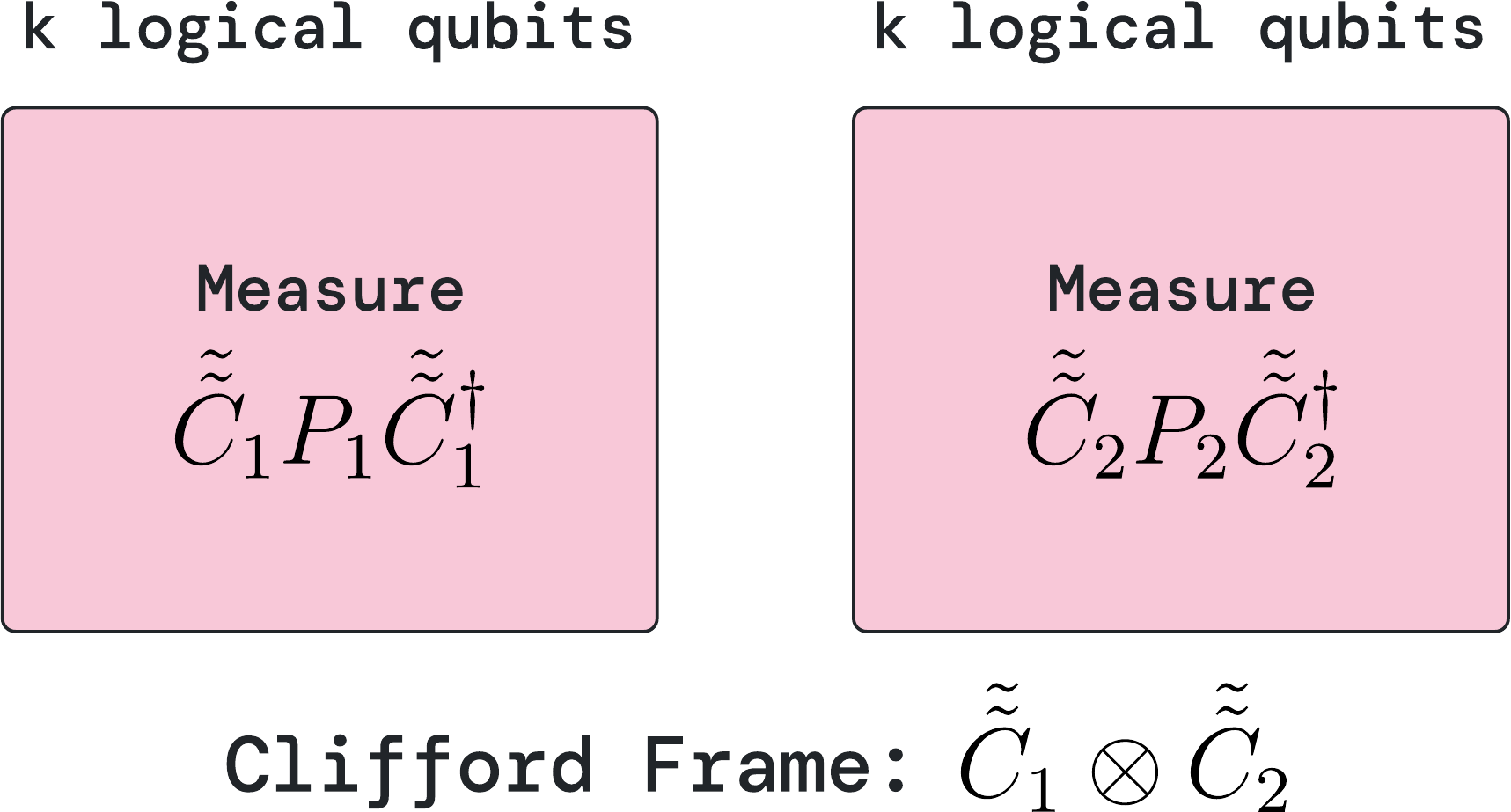}
}

\caption{Process of joining and separating processing units with Clifford frame cleaning to allow for flexible parallelism.}
\label{fig:cleaning}
\end{figure}

This framework allows for processing units to be joined during parts of a circuit in which many inter-unit gates occur, but then to be separated again (at a relatively small cost) for parts of the circuit that are more amenable to parallelisation.
The choice of if and when to separate units can be specifically optimised for compilation of any particular circuit, with the potential for significant time savings compared to either a fully serial approach or an approach that depends on physical implementation of all inter-unit entangling gates.

\subsubsection{General Operation}
\label{sec:general-operation}
The final step to our fully general operation is to optionally incorporate the memory.
Recall that each processing unit accesses memory via a port.
We allow for read-only memory access, which requires only gates that act as a control on the port and a target on the processing unit \cite{babbush_encoding_2018,gidney_how_2025}.
This means that the access can be provided by using logical CNOT gates with controls on the logical qubits in a port and targets on ancillary logical qubits in the processing unit to fan out memory data onto the processing unit at the start of the access and fan in at the end of the access.
Since such operations commute, arbitrarily many processing units can access the memory in parallel, provided they each have ports with measurement gadgets that can be used in parallel.

Implementing memory access on the Pinnacle Architecture uses the same concepts of joining and separating units presented in \cref{sec:parallel-operation}.
Specifically, when a processing unit accesses a window of the memory, the port associated with that window is joined onto the processing unit by the logical CNOTs which implement the fan-out.
Since no entangling gates act within the memory, these ports can always be assumed to be separate from one another.
Moreover, as all gates between a port and processing unit act as controls on the port, commuting through the Clifford frame only gives rise to $Z$-type logical measurements on the port.
When the access is finished, the Clifford frame is cleaned off the port such that the port can be separated from the processing unit again.
As shown in \cref{lem:cleaning2}, this requires only $2w$ logical cycles for a port of $w$ qubits.
This ensures that subsequent logical measurements on the processing unit act trivially on the memory logical qubits, enabling subsequent access by different processing units.

\subsection{Scalability}
\label{sec:features}
The architecture is designed to ensure that its operation remains feasible at large scale.
Specifically, because the processing units are assembled from individual processing blocks that are all connected via bridges between nearest neighbours, logical operators with support across any subset of processing blocks in a processing unit can be measured using connections between physical qubits that are restricted to the scale of one processing block.
This means that arbitrarily large processing units supporting arbitrarily many logical qubits can be realised with physical connections of constant scale.
For example, in a two-dimensional arrangement of qubits, this scale is approximately the square-root of the size of the processing block $\sqrt{n_{pb}}$.

This means that the architecture can be supported even on hardware platforms that only support interactions whose fidelity decreases continuously with the interaction distance.
Moreover, it means that no routing of logical qubits across the architecture is required.
Instead, all changes between logical cycles required to support different logical operations are confined to the scale of a processing block.
As discussed in \cref{sec:gb-codes}, code choice and gadget design can minimise the dynamism required even on this scale to be very limited.
In summary, the significant overhead reductions supported by QLDPC codes can be realised with interactions and rewiring confined to a fixed scale, which can be chosen to be consistent with a given hardware platform.

The further modularisation of the architecture into processing units provides an additional benefit. 
Specifically, feasibility limitations are expected to require large-scale quantum computers on many hardware platforms to be assembled from smaller modules~\cite{mohseni_how_2025, filippov_architecting_2025, yoder_tour_2025}.
This opens up an opportunity for hardware and architecture co-design that can be realised by associating these hardware modules with processing units of the Pinnacle Architecture.
This is beneficial because it aligns compilation imperatives with hardware constraints.
Gates between hardware modules should be minimised since they are likely to have poorer performance than intra-module gates, while gates between processing units should also be minimised to maximise parallelism.

\subsection{Structure}
The only constraint on the assembly of the architecture is that modules that are joined together at any time must be connected in the architecture. 
This ensures that these modules can be bridged, and therefore logical measurements are possible across them, without requiring long-distance transport.
In particular, this means that each processing unit must be adjacent to its associated magic engine and---if a memory is present---a port of memory. 
If a set of processing units is to be joined, these processing units should also be adjacent to each other.

The structure of the architecture for different example instantiations is shown in \cref{fig:architecture}.

\section{Instantiation of the Architecture}
\label{sec:instantiation}
In this section, we present a specific instantiation of the Pinnacle Architecture using a family of generalised bicycle (GB) codes, along with numerical simulation results used to determine the logical error rates that can be achieved for different of code distances and physical error rates.

\subsection{Setup}
\label{sec:gb-codes}
GB codes are defined by a lift $l\in\mathbb{N}$ and sets $A,B\subseteq\mathbb{Z}_l$~\cite{kovalev_quantum_2013}.
They have $n=2l$ physical qubits which can be divided into two sectors of $l$ physical qubits each, which we label $L$ and $R$.
The parity check operators of the code are then
\begin{align}
    S_{X,j} &=\prod_{a\in A} X_{(j+a),L} \prod_{b\in B} X_{(j+b),R}, \label{eq:gb-sx}\\
    S_{Z,j} &=\prod_{a\in A} Z_{(j-a),R} \prod_{b\in B} Z_{(j-b),L}. \label{eq:gb-sz}
\end{align}
Here, the first subscript on each operator denotes the position of a qubit within the sector and the second denotes the sector. 
For any $\sigma\in \mathbb{Z}_l$, a cyclic shift of all physical qubits by $\sigma$ sites preserves the group of parity check operators, making it a qubit automorphism.

We construct code blocks from the specific family of GB codes presented in Ref.~\cite{webster_explicit_2025}, which were first discovered in Ref.~\cite{panteleev_degenerate_2021} and subsequently explored in Refs.~\cite{crest_stabilizer_2023, lin_singleshot_2025, wang_coprime_2026}.
These codes have weight-six parity check operators and require only simple, relatively short-distance transport patterns for syndrome extraction.
The code family is parameterised by an integer $m>3$, and defined by choosing the lift to be $l=2^m-1$, as well as the sets $A$ and $B$ such that the polynomials $A(x) = \sum_{a \in A} x^a$ and $B(x) = \sum_{b \in B} x^b$ generate the parity check matrices of the classical simplex codes. 
Classical simplex codes have parameters $[2^m-1, m, 2^{m-1}]$~\cite{macwilliams_theory_1977}.
We conjecture that there exist choices of $A$ and $B$ such that the GB codes constructed in this way have parameters $\llbracket 2(2^m-1), 2m, m+(m-4)^2 \rrbracket$.
We present explicit instances for the first five codes in the family in \cref{tab:codes}.

For these codes, we empirically find that performance is improved by allowing for slightly more rounds of syndrome extraction than the code distance.
Guided by this observation, we choose to use $d_t=d+2$ code cycles per logical cycle.

The code blocks are extended to processing blocks by supplementing them with the gadget system presented in Ref.~\cite{webster_explicit_2025}.
This gadget system consists of four gadgets which correspond to four seed operators chosen such that all logical Pauli operators are products of cyclic shifts of the seed operators.
In particular, the $k$ logical qubits of each codes naturally divide into two logical sectors ($L$ and $R$) with $k/2$ logical qubits in each such that one $X$-type and one $Z$-type seed operator suffice for each logical sector.
Constructing, bridging, and shifting four gadgets capable of measuring each of the seed operators therefore suffices to measure arbitrary logical Pauli operators on a code block.
This means that only minor alterations are required to allow the same syndrome extraction circuit to measure any logical Pauli operator.
Four bridges are also included per processing block; three are used to bridge the four gadgets within the block, while the fourth is used to bridge the last gadget of the block to the first gadget of the next block.

Therefore, letting $n_g$ be the number of physical qubits per gadget and $n_b$ be the number of physical qubits per bridge, and accounting for $n_c=n$ check qubits, the total number of physical qubits per processing block is given by
\begin{equation}
    n_{pb}=n_{cb}+4n_g+4n_b.
\end{equation}
The values of these parameters are provided in \cref{tab:codes}.

\begin{table*}[]
    \caption{Instances of the family of generalised bicycle codes.
    For each code, the parameters $\llbracket n,k,d\rrbracket$ are provided, along with the number of code cycles per logical cycle, $d_t=d+2$.
    The codes are defined by $l$, $A$, and $B$, alongside \cref{eq:gb-sx} and \cref{eq:gb-sz}.
    The remaining columns show the number of physical qubits in their code blocks, gadgets, bridges and processing blocks.
    See Ref.~\cite{webster_explicit_2025} for more details.}
    \label{tab:codes}
    \begingroup
    \renewcommand{\arraystretch}{1.1}
    \begin{tabular}{| c | c | c | c | c |c | c | c | c |}
        \hline
        $\llbracket n,k,d\rrbracket$ & $d_t$ & $l$ & $A$ & $B$ & Code Block Qubits & Gadget Qubits & Bridge Qubits & Processing Block Qubits\\
        & $(d+2)$ & & & & ($n_{cb}=2n$) & ($n_g$) & ($n_b$) & ($n_{pb}=n_{cb}+4n_g+4n_b$)\\
        \hline
        $\llbracket 30,8,4\rrbracket$ & $6$ & $15$ & $\{0,6,13\}$ & $\{0,1,4\}$  & $60$ & $13$ & $7$ & $140$ \\
        \hline
        $\llbracket 62,10,6\rrbracket$ & $8$ & $31$ & $\{0,6,15\}$ & $\{0,5,7\}$ &  $124$ & $19$ & $11$ & $244$\\
        \hline
        $\llbracket 126,12,10\rrbracket$ & $12$ & $63$ & $\{0,4,37\}$ & $\{0,29,49\}$ & $252$ & $31$ & $19$ & $452$\\
        \hline
        $\llbracket 254,14,16\rrbracket$ & $18$ & $127$ & $\{0,32,100\}$ & $\{0,28,49\}$ & $508$ & $57$ & $31$ & $860$\\
        \hline
        $\llbracket 510,16,24\rrbracket$ & $26$ & $255$ & $\{0,39,55\}$ & $\{0,70,127\}$ & $1020$ & $99$ & $51$ & $1620$\\
        \hline
    \end{tabular}
    \endgroup
\end{table*}

These gadgets can also be used to measure certain sets of logical operators in parallel (i.e., in a single logical cycle) via the inclusion of duplicate gadgets.
Specifically, a set of $m$ logical operators $P_1,\ldots P_m$ can be measured in parallel by using a different copy of the gadget to measure each, provided they commute on every physical qubit. 
This qubit-wise commutation condition ensures that the check operators from different gadgets all commute.
Connecting $m$ gadgets to the same code block can increase code check weights and qubit degrees in places where multiple gadgets are joined to the same qubit or check.
In theory, this increase can be by up to $m$ but, with in an appropriately chosen basis, it is typically significantly less than this.
If necessary, a small number of ancillary qubits can also be used to reduce check weights and/or qubit degrees~\cite{williamson_lowoverhead_2024, cowtan_parallel_2026}.

\subsection{Modules}
\label{sec:instantiation-components}

\subsubsection{Processing Units}
Using $\beta$ processing blocks constructed from the GB code family introduced above (for any $\beta\in\mathbb{N}$), we can encode $\kappa=\beta k$ logical qubits in $\beta n_{pb}$ physical qubits.
Specifically, with a code distance of $d=16$, we can encode $14\beta$ logical qubits in $860\beta$ physical qubits.
For better protection, we can instead use a code distance of $d=24$ and encode $16\beta$ logical qubits in $1620\beta$ physical qubits.

\subsubsection{Magic Engines}
\label{sec:gb-me}
We construct each magic engine from code blocks of the same GB code family as those used for the processing blocks.
These blocks naturally have the required $L$ and $R$ logical sectors, with $k/2> 5$ logical qubits in each sector when $d\geq 10$.

We use 15-to-1 magic state distillation on these code blocks to produce encoded $\ket{\bar{T}}$ magic states~\cite{bravyi_universal_2005}. 
Following Ref.~\cite{litinski_magic_2019}, this can be done using fifteen $Z$-type $\pi/8$ rotations, followed by four logical measurements used for post-selection.
Each of the rotations is implemented by injecting a noisy $\ket{T}$ state from a small ancillary $\llbracket n_a,1,d_a\rrbracket$ code.
This is done using generalised surgery to perform a joint measurement (repeated $d_a$ times) between the $L$ logical sector of the GB code and the ancillary code.
The noisy $\ket{T}$ states are prepared using standard techniques (including, if necessary, methods to increase the input state fidelity, such as zero-level distillation~\cite{goto_minimizing_2016, itogawa_efficient_2025} or magic state cultivation~\cite{gidney_magic_2024, sahay_foldtransversal_2025}).
Four post-selection measurements (each repeated $r$ times) are then performed in series on the $L$ logical sector, and the state is rejected if any round of any of these measurements gives a $-1$ result.

While in principle the fifteen rotations could be implemented in parallel, in practice it is preferable to split them into two batches so that at most eight states are injected in parallel.
Accounting also for the measurements required for injecting the high-fidelity state into the processing unit, this means that a maximum of ten logical measurements are required in parallel.
Since these operators are fixed, a basis can be chosen that limits their overlap to ensure that these measurements do not result in a large increase in qubit degrees of check weights.

The physical qubits that must be accounted for in this protocol are GB code block itself ($n_{cb}$), gadgets used to perform up to ten logical measurements in parallel ($10n_g$), fifteen ancillary $\llbracket n_a,1,d_a\rrbracket$ codes ($15\times (2n_a-1)$) with bridges used to connect these ancillary codes to the GB code block ($15\times (2d-1$)) and (if necessary) ancillary qubits ($n_\alpha$) used to reduce the infidelity of input states (e.g., by magic state cultivation).
The total overhead of the magic engine is therefore given by
\begin{equation}\label{eq:magic-engine}
n_{me}=n_{cb}+10n_g+30(n_a+d_a-1)+n_\alpha.
\end{equation}

The parameters of the engine are chosen to achieve a target output infidelity, $p_{\text{out}}$.
Specifically, we choose the distance of the GB code block, $d_e$, to be large enough that the logical error rate of this block is much smaller than $p_{\text{out}}$, which ensures that the contribution to the output infidelity from errors on this code block is negligible.
The probability of an error in one of the $Z$-type rotations can then be approximated by $p_{\text{rot}}\approx p_{\text{in}}+(d_a+1)p_a$, where $p_{\text{in}}$ is the error of the noisy $|T\rangle$ state, $p_a$ is the logical error rate per code cycle of the ancillary code and the coefficient $d_a+1$ accounts for $d_a$ rounds of joint measurements with the GB block and one destructive measurement on the ancillary code.
The probability of an undetected error in the post-selection measurements is $p_m=p^r$.
In the 15-to-1 distillation scheme, there are thirty-five failure modes arising from errors in $Z$-type rotations \cite{litinski_magic_2019}, and a further six resulting from a rotation error combined with two errors in the post-selection measurements.
The infidelity of the output state can therefore be approximated by
\begin{align}
p_{\text{out}}&\approx 35p_{\text{rot}}^3+6p_{\text{rot}}p_m^2 \nonumber\\
&\approx 35\left(p_{\text{in}}+(d_a+1)p_a\right)^3+6p^r\left(p_{\text{in}}+(d_a+1)p_a\right) \label{eq:engine-condition}
\end{align}
The reject rate is approximately
\begin{equation}
p_r \approx 15p_{\text{rot}}+4rp_m+4p_\alpha
\end{equation}
where $p_\alpha$ is the probability that preparation of sufficiently many noisy $|T\rangle$ states for the protocol fails (e.g., because of post-selection in magic state cultivation).

Motivated by the applications of \cref{sec:applications}, we consider two target output infidelities: $p_{\text{out}}\leq 10^{-9}$ (suitable for the Fermi-Hubbard model) and $p_{\text{out}}\leq 10^{-11}$ (suitable for RSA-2048 factoring).
We also consider two physical error rates, $p=10^{-4}$ and $p=10^{-3}$; motivated by the results of \cref{sec:performance} we choose $d_e=10$ and $d_e=24$ respectively for these error rates. 
In the case where $p_{\text{out}}= 10^{-9}$ and $p=10^{-4}$, it is sufficient to use a physical $|T\rangle$ state (i.e., to choose the trivial code where $n_a=d_a=1$) and $r=1$.
Indeed, in this case $p_{\text{in}}=p_a=p=10^{-4}$ so that $p_{\text{out}}\approx 41(3\times 10^{-4})^3=10^{-9}$.
The number of qubits required by the magic engine is then
\begin{equation}
n_{\text{me; $10^{-4} \to 10^{-9}$}}=592.
\end{equation}
The reject rate is given by $p_r\approx 19p\approx 0.2\%$.

For all other cases, we use a version of fold-transversal cultivation \cite{sahay_foldtransversal_2025}.
Specifically, in each case we perform the injection and cultivation steps in one or more ancillary systems of 25 qubits per injected state, which each produce a magic state in no more than eight code cycles.
These states can then be grown into fifteen distance-$d_a$ rotated surface codes of required distance $d_a$ and post-selected to have an infidelity of $p_{\text{in}}$ before injection into the QLDPC code block.
The magic engine reject rate is dominated by this post-selection in cultivation (i.e., $p_r\approx p_\alpha$) \footnote{We estimate these reject rates using data on expected attempts (which we denote $\gamma$) from Ref.~\cite{sahay_foldtransversal_2025}, extrapolating under the assumption that the failure rate is linear in $p$ for the case of $p=10^{-4}$.
Specifically, the overall reject rate is then given by $p_r=(1-(1-1/\gamma)^\xi)^{15}$, where $\xi$ is the number of attempts that are possible in the available time and with the throughput provided by the ancilla systems ($\xi=2$ for the $p=10^{-4}$ case and $\xi=5$ for the $p=10^{-3}$ cases).}.
For $p=10^{-4}$ and $p_{\text{out}}=10^{-11}$, we choose $d_a=5$ (such that $p_a=5\times 10^{-7}$  \cite{orourke_compare_2025}), $p_{\text{in}}=5\times 10^{-5}$ and $r=1$, and we provision one ancilla system per injected state.
The number of physical qubits required is then:
\begin{equation}
n_{\text{me; $10^{-4} \to 10^{-11}$}}=1807.
\end{equation}
and the estimated reject rate is $p_r\approx 2\%$.
For $p=10^{-3}$ and $p_{\text{out}}=10^{-9}$, we choose $d_a=7$ (such that $p_a=2\times 10^{-5}$  \cite{orourke_compare_2025}), $p_{\text{in}}=10^{-4}$ and $r=2$, and we provision two ancilla systems per injected state.
The number of physical qubits required is then:
\begin{equation}
n_{\text{me; $10^{-3} \to 10^{-9}$}}=4410.
\end{equation}
For $p=10^{-3}$ and $p_{\text{out}}=10^{-11}$, we choose $d_a=9$ (such that $p_a=2\times 10^{-6}$  \cite{orourke_compare_2025}), $p_{\text{in}}=5\times 10^{-5}$ and $r=2$, and we provision two ancilla systems per injected state. 
The number of physical qubits required is then:
\begin{equation}
n_{\text{me; $10^{-3} \to 10^{-9}$}}=5430.
\end{equation}
For both $p=10^{-3}$ cases, the estimated reject rate is $p_r\approx 10\%$.

Since the engine is intended to produce a magic state for each logical cycle, the time, $t_{\text{me}}$, required for distillation places a lower bound on the logical cycle time of the associated processing unit.
The time required for the logical measurements in distillation themselves is $(2d_a+4r)t_c$.
However, in some cases this scheme can be reaction-limited, since the post-selection measurement bases depend on the outcome of the injection measurements.
Specifically, we must allow at least the reaction time of $t_r=10t_c$ from the beginning of the protocol to the first post-selection measurement and from the beginning of the second injection batch to the second post-selection measurement \footnote{This suffices since we can choose the eight logical measurements in the first batch to act trivially on one of the four logical qubits involved in post-selection.}.
This means the number of code cycles required is:
\begin{equation}
t_{me}=\max\{2d_at_c+4rt_c,t_r+4rt_c,d_at_c+t_r+3rt_c\}
\end{equation}
For the $p=10^{-3}$ and $p=10^{-4}$ magic engines, we find that $t_{me}\leq 26t_c$ and $t_{me}\leq 18t_c$ respectively, which ensures that distillation completes within $d=24$ and $d=16$ GB code logical cycles.
This ensures compatibility for the applications in \cref{sec:applications} without any increase in logical cycle time, with the sole exception of the Fermi-Hubbard model with $p=10^{-4}$.
In that case, the use of $d=10$ processing blocks with the $p_{\text{out}}=10^{-9}$ magic engine implies a modest increase in logical cycle time from $d_t=12t_c$ to $t_{me}=14t_c$ code cycles, which is accounted for in the presented runtimes.

\subsubsection{Memory}\label{sec:gb-memory}
For the memory, we use the same code blocks as are used for the processing blocks.
For simplicity, we match the window size with the number of logical qubits in a logical sector, $k/2$.
Each port then corresponds to one of the $Z$-type gadgets used in the gadget system of the processing blocks, along with a bridge to connect to a processing unit.
The logical operators that must be measured to access the memory commute on every physical qubit and so can be measured in parallel, as they act across disjoint processing units and act only as $Z$-type operators on the memory.
Moreover, since each memory code block has at most two ports, these measurements increase the check weight and qubit degree by at most two.

Referring to \cref{tab:codes}, the additional number of physical qubits per port is $n_g+n_b=88$ at $d=16$ or $n_g+n_b=150$ at $d=24$.
Hence, for any $\nu\in\mathbb{N}$, we can encode $14\nu$ logical qubits in memory such that $\rho$ processing units can access it in parallel with $508\nu+88\rho$ physical qubits at $d=16$ or $16\nu$ logical qubits in memory with $1020\nu+150\rho$ physical qubits at $d=24$.

\subsection{Simulation Results}
\label{sec:performance}
To assess the logical error rates achievable with different code choices, we perform numerical simulations of both memory and logical measurements by generalised surgery.
These are performed for one logical cycle ($d_t=d+2$ rounds of syndrome extraction) with standard circuit-level depolarising noise.
For memory experiments we instead simulated $d$ rounds of syndrome extraction and rescaled the logical failure rates by $(d+2)/d$.
Circuits for the memory simulations are constructed as in Ref.~\cite{lin_singleshot_2025}, while circuits for the surgery simulations are constructed using integer linear programming~\cite{vittal_flagproxy_2024}.
We perform uncorrelated (only $X$-type detectors) most-likely error decoding by converting the decoding problem into a mixed integer program and allow the solver to obtain an optimal solution.
Results of these simulations are shown in \cref{fig:ler}.

\begin{figure}[]
    \centering
    \includegraphics[width=\linewidth]{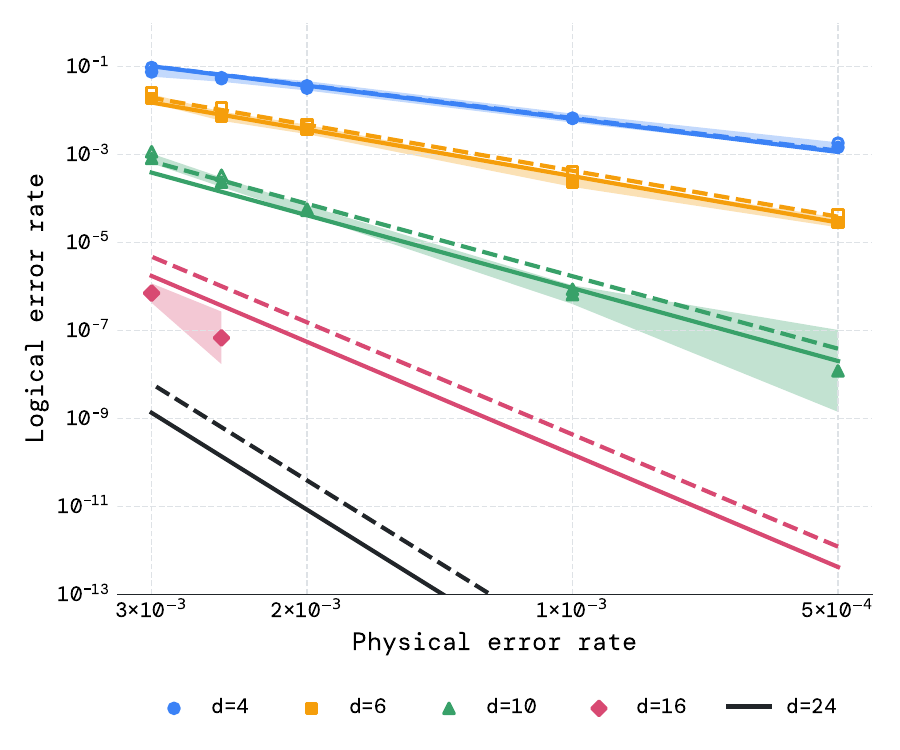}
    \caption{
    Simulation results to determine logical error rates per processing block (of $k$ logical qubits) per logical cycle of the Pinnacle Architecture using GB codes of different distances and across different physical error rates.
    Solid markers show results for $X$-basis memory experiments decoded with most-likely error decoding.
    Hollow markers show results for logical measurements by generalised surgery; data was collected for the $d=4$, $d=6$ and $d=10$ codes at the same physical error rates as for the memory experiments.
    Highlighted regions are $99\%$ confidence intervals for the memory experiment data points.
    The solid and dashed lines are fits of the ansatz in \cref{eq:ansatz} to the memory experiment data and logical measurement data respectively.
    We emphasise that ansatz parameters are independent of distance (i.e., the same equations are used for all codes), which allows for extrapolation of the collected data across the code family.
    }
    \label{fig:ler}
\end{figure}

We assume each point is binomially distributed and use maximum likelihood estimation to fit the points to the sub-threshold ansatz
\begin{equation}\label{eq:ansatz}
    p_{L,cb}(p, d) = A \left(\frac{p}{B}\right)^{\frac{d}{2} + C}. 
\end{equation}
Here $p_L$ is the total logical failure rate for all $k$ logical observables over $d_t$ rounds.
It follows from this that the error rate per logical observable and logical cycle is given by
\begin{equation}\label{eq:ansatz-per-qubit}
    p_L(p, k, d) = \frac{A}{k} \left(\frac{p}{B}\right)^{\frac{d}{2} + C}.
\end{equation}
The fitted parameters for the ansatz with $95\%$ confidence intervals are given in \cref{tab:ansatz-fits}.
\Cref{tab:ler} shows the logical error rates per logical qubit per logical cycle derived from this ansatz for each of the GB codes considered at physical error rates of $p=10^{-3}$ and $p=10^{-4}$.

\begin{table}
    \centering
    \caption{Fitted parameters with $95\%$ confidence intervals for the ansatz in Eq.~\ref{eq:ansatz} for memory experiments and logical measurement experiments.}
    \label{tab:ansatz-fits}
    \begingroup
    \renewcommand{\arraystretch}{1.25}
    \begin{tabular}{|c|c|c|c|}
        \hline
        Experiment & $A$ & $B$ & $C$ \\
        \hline
        Memory & $5.9^{+1.8}_{-1.4}$ & $0.0179^{+0.0006}_{-0.0006}$ & $0.50^{+0.09}_{-0.09}$ \\
        \hline
        Log. Meas. & $6.2^{+1.9}_{-1.4}$ & $0.0158^{+0.0007}_{-0.0007}$ & $0.47^{+0.09}_{-0.09}$ \\
        \hline
    \end{tabular}
    \endgroup
\end{table}

We emphasise that the purpose of these simulation results is to benchmark the capabilities of the architecture and guide code distance choice for resource estimation.
This informs the decision to use most-likely error decoding which correctly decodes all faults of weight less than $d/2$ and avoids error floors that can arise from alternatives such as belief propagation decoders~\cite{raveendran_trapping_2021}.
The problem of developing a sufficiently fast decoder for real-time use by the classical control system of quantum hardware is outside the scope of this paper, and we look forward to addressing it in future work.

\begin{table}
    \centering
    \caption{Error rates per logical qubit and logical cycle for logical measurement in GB codes of each distance at physical error rates of $p=10^{-3}$ and $p=10^{-4}$.
    Values correspond to the central estimates of the fit parameters from the Logical Measurement row of \cref{tab:ansatz-fits} substituted into the ansatz presented in \cref{eq:ansatz-per-qubit}.}
    \label{tab:ler}
    \begingroup
    \renewcommand{\arraystretch}{1.25}
    \begin{tabular}{|c|c|c|c|c|c|}
        \hline
        $p$ & $d=4$ & $d=6$ & $d=10$ & $d=16$ & $d=24$\\
        \hline
        $10^{-3}$ & $8\times 10^{-4}$ & $4\times 10^{-5}$ & $1\times 10^{-7}$ & $3\times 10^{-11}$ & $4\times 10^{-16}$\\
        \hline
        $10^{-4}$ & $3\times 10^{-6}$ & $1\times 10^{-8}$ & $5\times 10^{-13}$ & $1\times 10^{-19}$ & $1\times 10^{-28}$\\
        \hline
    \end{tabular}
    \endgroup
\end{table}

\section{Applications}
\label{sec:applications}
In this section, we show how the Pinnacle Architecture can be applied to two applications: determining the ground state energy of the Fermi-Hubbard model, and factoring RSA integers.

\subsection{Fermi-Hubbard Model}
\label{sec:fh}
In this subsection, we determine the resources required to estimate the ground state energy of the two-dimensional Fermi-Hubbard model using the Pinnacle Architecture. 

\subsubsection{Algorithm}
The two-dimensional Fermi-Hubbard model represents a system of interacting fermions and has the Hamiltonian
\begin{align}
    H &= H_h + H_I \notag \\
    &= \sum_{\langle i,j\rangle} \sum_{\sigma\in \{\uparrow,\downarrow\}} \left(a^\dag_{i,\sigma} a_{j,\sigma} +  a^\dag_{j,\sigma} a_{i,\sigma}\right) + u\sum_i \hat{n}_{i,\uparrow}\hat{n}_{i,\downarrow}.
\end{align}
Here, $i$ denotes the sites of an $L\times L$ lattice, $\langle i,j\rangle$ denotes pairs of nearest neighbours on this lattice, $\sigma\in \{\uparrow,\downarrow\}$ denotes spins states, $a^\dag$ and $a$ represent creation and annihilation operators, respectively, $\hat{n}=a^\dag a$ denotes the number operator, and $u$ denotes the (dimensionless) coupling strength.

We follow Ref.~\cite{campbell_early_2022} in using plaquette Trotterisation to determine the ground state energy of the Fermi-Hubbard model with a relative error of 0.5\% of the total lattice energy.
In order to minimise the number of qubits required, we modify this method by omitting Hamming weight phasing (which requires a system of $1\leq\alpha\leq L^2/2$ ancillary logical qubits), instead simply performing each required phase rotation in sequence.
With this simplification, the number of logical qubits required is
\begin{equation}\label{eq:FH-N-formula}
    N=2L^2+2.
\end{equation}
This accounts for two logical qubits for each lattice site (one for each spin state) and two additional logical ancilla qubits (one for phase estimation and one for repeat until success synthesis)~\cite{campbell_early_2022}.
From Eq.~(F10) of Ref.~\cite{campbell_early_2022}, the $T$ count for one shot of the algorithm is given by
\begin{equation}\begin{split}
    \tau&=6.203\sqrt{\frac{W}{\left(\epsilon(1-x)\right)^3}}\,\times\\
    &\left(N_R\left(1.15\log_2\left(\frac{N_R\sqrt{3W}}{x\sqrt{1-x}\sqrt{\epsilon^3}}\right)+9.2\right)+N_T\right).
\end{split}\end{equation}
Here, $N_T=12L^2$ and $N_R=4L^2$ are the number of $T$ gates and arbitrary $Z$ rotations (which are each synthesised from a number of $\bar{T}$ gates) required per Trotter step.
The number of Trotter steps required for the target precision is given by the prefactor.
In this prefactor, $\epsilon=0.005E_0$ corresponds to the required relative error of 0.5\%, where the energy per site $E_0$ is estimated at 1.02 hartrees for $u=4$ and 0.74 hartrees for $u=8$~\cite{kivlichan_improved_2020}.
$W$ is a parameter that bounds the Trotter error; bounds on this parameter are provided in Ref.~\cite{campbell_early_2022}.
The parameter $x$ reflects a choice of how to split the error budget across different sources; we optimise over this in our analysis.

In addition to T gates, the algorithm also requires additional logical measurements.
In particular, repeat-until-success synthesis is used to perform each of the $N_R$ arbitrary $Z$ rotations in each Trotter step~\cite{paetznick_repeatuntilsuccess_2014, bocharov_efficient_2015}\footnote{For simplicity, we here follow Ref.~\cite{campbell_early_2022} in using this approach. We note that more recent advances in rotation synthesis have since been presented \cite{kliuchnikov_shorter_2023}, which may allow for reduced runtimes if incorporated in future.}.
This method successfully performs the correct rotation with a probability of at least one half~\cite{bocharov_efficient_2015}, and so the expected number of attempts per rotation is at most two.
Since one logical measurement is required per attempt (in addition to the $T$ gates), this adds an average of up to two logical measurements per arbitrary rotation.
Logical measurements are also required for the final phase estimation.
However, since this is only performed once (not for each Trotter step), it is negligible compared with the $T$ count.
We therefore express $\mathcal{T}$---the sum of the $T$ count and number of logical measurements---as
\begin{equation}\begin{split}
    \mathcal{T}&=6.203\sqrt{\frac{W}{\left(\epsilon(1-x)\right)^3}}\,\times\\
    &\left(N_R\left(1.15\log_2\left(\frac{N_R\sqrt{3W}}{x\sqrt{1-x}\sqrt{\epsilon^3}}\right)+11.2\right)+N_T\right).
\end{split}\end{equation}
This quantity determines the number of logical cycles on the Pinnacle Architecture.
It is approximately constant in $L$ because the allowed error is relative to the total energy, which scales as $L^2$~\cite{kivlichan_improved_2020}.

\subsubsection{Implementation and Results}
Concretely, we consider the case of even $L\leq 32$ and $u=4$.
In this regime, $N\leq 2050$ and we find numerically that the number of logical cycles satisfies $\mathcal{T}=8\times 10^6$, which also upper bounds the $T$ count.
This implies that the logical spacetime volume satisfies $N\mathcal{T}\leq 2\times 10^{10}$.
Hence, the algorithm can be implemented with negligible failure probability provided the error rate per logical qubit and logical cycle satisfies $p_L \ll 5\times 10^{-11}$ and the $\ket{T}$ state fidelity satisfies $p_T\ll 10^{-7}$.
With reference to \cref{tab:ler}, this is satisfied by using the $d=24$ GB code instantiation of the architecture for a physical error rate of $p=10^{-3}$ and the $d=10$ instantiation for $p=10^{-4}$, along with the corresponding $p_{\text{out}}=10^{-9}$ magic engines.

To implement this on the Pinnacle Architecture, we use a single processing unit with $\beta=\left\lceil N/k\right\rceil$ processing blocks to ensure there are $\kappa\geq N$ logical qubits available.
We also account for a magic engine for this processing unit, but do not include memory.
The required number of physical qubits is therefore given by
\begin{equation}
    n=n_{pb}\left\lceil\frac{N}{k}\right\rceil+n_{me}.
\end{equation}
Substituting $N=2L^2+2$ and the values for $n_{pb}$, $n_{me}$ and $k$ from \cref{sec:instantiation-components}, the numbers of physical qubits required with physical error rates of $p=10^{-3}$ and $p=10^{-4}$, respectively, are therefore
\begin{align}
    n_{10^{-3}} &= 1620\left\lceil \frac{L^2+1}{8} \right\rceil+4410,\\
    n_{10^{-4}} &= 452\left\lceil \frac{L^2+1}{6} \right\rceil+1807.
\end{align}
\Cref{tab:FH-results} shows the number of physical qubits required for even $L\leq 32$, along with the runtimes with code cycle times of 1 \textmu s and 1 ms.
\Cref{fig:FH-graph} shows that the number of physical qubits required is an order of magnitude smaller than those required using surface codes in Ref.~\cite{kivlichan_improved_2020}.

\begin{table}
    \centering
    \caption{Physical qubits and runtime required to perform one shot of Fermi-Hubbard ground state energy estimation on an $L\times L$ lattice with a relative error of $\leq 0.5\%$ of the total lattice energy. 
    The runtime is approximately independent of $L$ because the relative error allows for fewer Trotter steps for larger lattices; runtimes for each value of $L$ are equal to or slightly smaller than those given. 
    kq represents kiloqubits ($\times 10^3$ qubits).}
    \label{tab:FH-results}
    \begin{tabular}{|c|c|c|}
        \hline
        \multicolumn{3}{|c|}{Physical Qubits}\\
        $L$ & $p=10^{-3}$ & $p=10^{-4}$\\
        \hline 
        8 & 19 kq & 5.6 kq\\
        10 & 25 kq & 8.3 kq\\
        12 & 35 kq & 12 kq\\
        14 & 45 kq & 16 kq\\
        16 & 58 kq & 20 kq\\
        18 & 71 kq & 25 kq\\
        20 & 87 kq & 31 kq\\
        22 & 103 kq & 37 kq\\
        24 & 123 kq & 44 kq\\
        26 & 142 kq & 52 kq\\
        28 & 165 kq & 60 kq\\
        30 & 187 kq & 69 kq\\
        32 & 213 kq & 78 kq\\
        \hline
        \multicolumn{3}{|c|}{Runtime (Upper Bound)}\\
        Code Cycle & $p=10^{-3}$ & $p=10^{-4}$\\
        \hline
        1 \textmu s & 3.8 min & 1.8 min\\
        1 ms & 2.6 days & 1.3 days\\
        \hline
    \end{tabular}
\end{table}

A full determination of the ground state energy requires multiple shots of the algorithm, with the number dependent on the overlap between the initial state and the true ground state.
These shots can be performed in series (with a proportional increase in run time) or could be performed in parallel using multiple separate processing units (with a proportional increase in the required physical qubits).

\subsection{RSA Factoring}
In this subsection, we show that the Pinnacle Architecture can be used to efficiently perform the factoring necessary to break RSA encryption.

\subsubsection{Algorithm}
The algorithm we use is a generalisation of that presented by Gidney in Ref.~\cite{gidney_how_2025}, which uses techniques developed by Ekerå and Håstad~\cite{ekera_quantum_2017} and by Chevignard et al.~\cite{chevignard_reducing_2025}.
We refer to this algorithm as Gidney's algorithm. 
This algorithm uses residue number system arithmetic to replace modular arithmetic over $N_\mathrm{RSA}$ (the number being factored) with modular arithmetic over a set of primes $P$ that each have size polylogarithmic in $N_\mathrm{RSA}$.
This reduces the number of qubits required for the working register required for modular exponentiation (the dominant part of the factoring algorithm) from $\Theta\left(\log N_\mathrm{RSA}\right)$ to $\Theta\left(\log\log N_\mathrm{RSA}\right)$, reducing the space overhead.
For each individual prime, the time overhead for the modular exponentiation is also significantly reduced since the time cost for arithmetic operations such as addition scales with the size of the input registers. 
However, in Gidney's algorithm, this does not translate to an overall reduction in runtime because the $|P|$ primes are processed in series.

We generalise Gidney's algorithm by allowing for the possibility of processing multiple primes in parallel. 
This is done by adding $\rho-1$ additional ancillary working registers, for any positive integer $\rho\leq |P|$.
The outer loop of Gidney's algorithm can then be parallelised across the $\rho$ working registers, allowing the $|P|$ primes to be processed in $\lceil|P|/\rho\rceil$ batches.
Between computation and uncomputation of each batch, parallel reduction is used to combine the accumulators of each register onto the  accumulator of the first working register by aggregating accumulators pairwise in the form a binary tree.
This ensures that the final value of this accumulator matches that of Gidney's algorithm, since it simply amounts to a reordering of the sum used to calculate the approximate modular exponential (Eq.~(20) of Ref.~\cite{gidney_how_2025}), while the other working registers are fully uncomputed after each batch is processed to ensure that they end in the trivial state.
The required result may therefore be extracted by measurement and classical processing in the same way.
Letting $\mathcal{T}_G$ be the time cost for Gidney's algorithm, our algorithm has a time cost of
\begin{equation}
    \mathcal{T}=\left\lceil\frac{|P|}{\rho}\right\rceil\left(\frac{1}{|P|}\mathcal{T}_G+O\left(\log \rho\right)\right),
\end{equation}
where the $O\left(\log \rho\right)$ term accounts for cost associated with combining accumulators.
Since $|P|$ is large (e.g., $\approx 2.1\times 10^4$ for RSA-2048~\cite{gidney_how_2025}), this can allow for a reduction in the time cost by many orders of magnitude.

Importantly, the additional space incurred by this form of parallelisation is not proportional to the reduction in time cost.
To see why, note that the full register in Gidney's algorithm has two components --- an input register of $m=\Theta\left(\log N_\mathrm{RSA}\right)$ logical qubits and a working register of $N_w=\Theta\left(\log\log N_\mathrm{RSA}\right)$ logical qubits used to perform the approximate modular exponentiation.
While the working register must be duplicated to allow for parallelisation, the same input register can be reused for many primes as it is only accessed by lookup operations.
These operations commute for different primes since they all use only gates with controls on the input qubits.
More precisely, given that a working register accesses the input register in windows of $w_1$ logical qubits at a time, pipelining these accesses can allow $\lceil m/w_1\rceil$ registers to run in parallel using a single input register.
Therefore, the required number of logical qubits is
\begin{equation}\label{eq:N}
    N=\left\lceil\frac{\rho}{\lceil m/w_1\rceil}\right\rceil m + \rho N_w.
\end{equation}
Noting that $\left\lceil\frac{\rho}{\lceil m/w_1\rceil}\right\rceil=1$ for $\rho\lesssim 200$ and $m$ is an order of magnitude larger than $N_w$, this can be significantly smaller than $\rho N_1=\rho(m+N_w)$ (where $N_1$ is the number of logical qubits required for Gidney's algorithm). 
This leads to significant spacetime savings from parallelisation, as shown in \cref{fig:logical-spacetime}.

\begin{figure}[]
    \centering
    \includegraphics[width=\linewidth]{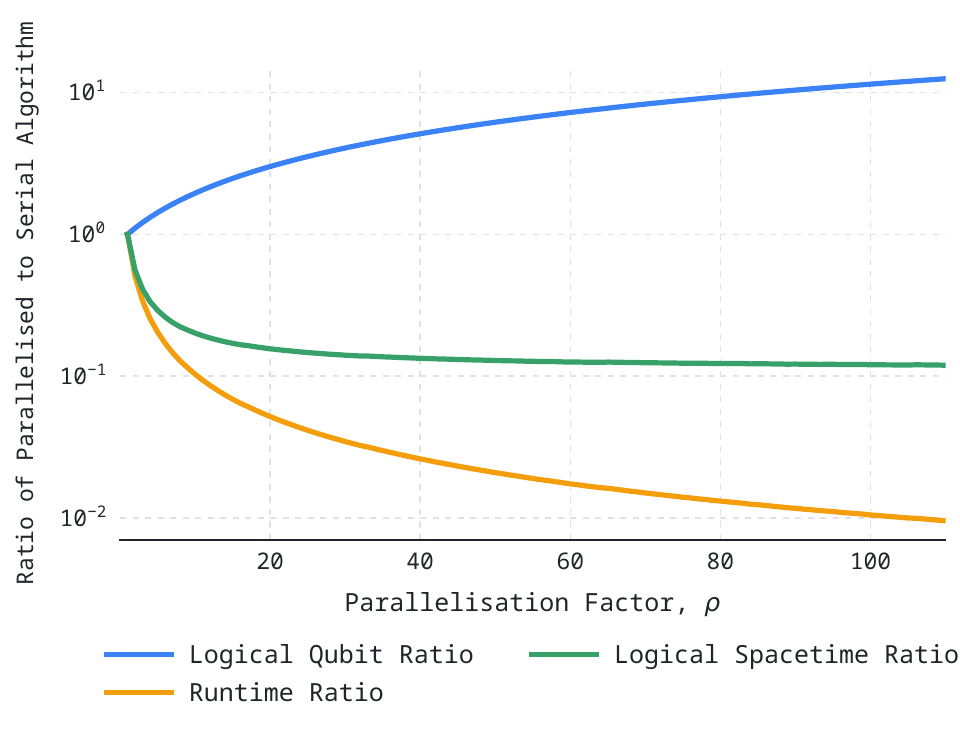}
    \caption{Comparison of space, time and spacetime cost for our parallelised algorithm to Gidney's algorithm or Ref.~\cite{gidney_how_2025}. 
    For the purpose of this plot, all algorithmic parameters (other than $\rho$) are chosen to match those of the $n=2048$ row of Table 5 of Ref.~\cite{gidney_how_2025}.}
    \label{fig:logical-spacetime}
\end{figure}

Our modified algorithm allows for very efficient parallelism, with orders of magnitude reductions in the time overhead achievable with a smaller increase in the space overhead.
In particular, we note that while Gidney's algorithm has a longer runtime than the earlier implementation of Gidney and Ekerå~\cite{gidney_how_2021}, our parallelised version can achieve a significantly faster runtime.
This motivates our choice of this algorithm even for architectures with slow clock cycles, where runtime is especially important.

\subsubsection{Implementation on Pinnacle Architecture}
\label{sec:implementation}
To implement the algorithm on the Pinnacle Architecture, we begin by allocating a processing unit for each working register.
These processing units can run in parallel throughout most of the computation.
The only exception is the relatively short periods when the accumulators of the working registers are  being aggregated; during this process, Clifford frame cleaning is used following each pairwise interaction of accumulators to ensure that the processing units do not remain joined thereafter.

Each of these working registers is equipped with enough logical qubits to allow implementation of Gidney's algorithm.
This corresponds to: an $f$ logical qubit sub-register for the overall accumulator; an $\ell+\mathrm{len}{(m)}$ logical qubit sub-register on which discrete-log values are accumulated for each prime; two ancillary sub-registers each with $\max{\left(f,\ell+\mathrm{len}{(m)}\right)}$ logical qubits; a third ancillary sub-register with $\ell$ logical qubits; and one additional ancillary logical qubit for compiling Toffoli gates from T gates using the circuit of Ref.~\cite{jones_lowoverhead_2013}.
As in Ref.~\cite{gidney_how_2025}, $f$ is the length of the truncated accumulator, $\ell$ is the bit length of the residue primes and $\mathrm{len}{(m)}=\lfloor \log_2(m)\rfloor +1$ is the bit length of $m$.
We therefore require $\kappa$ logical qubits, given by
\begin{equation}\label{eq:kappa}
    \kappa=f+2\ell+\mathrm{len}{(m)}+2\max(f,\ell+\mathrm{len}{(m)})+1.
\end{equation}
To achieve this, we allocate $\lceil\kappa/k\rceil$ processing blocks to each processing unit, where $k$ is the number of logical qubits per code block.

The input register (or multiple input registers if $\rho\geq\lceil m/w_1\rceil$) is associated with memory in the architecture.
Access to this memory occurs during loop 1 of Gidney's algorithm, when windows of $w_1$ logical qubits are used as the address for a lookup operation targeted on each working register.
As discussed in \cref{sec:gb-memory}, we fix $w_1=k/2$, and we also enforce the condition that $\ell\geq w_1$.
This ensures that there are at least $w_1$ unused logical qubits onto which the window can be fanned out, since at the time of this loop, $f+3\ell+3\,\mathrm{len}{(m)}$ logical qubits of each working register are in use, meaning that there are at least $\ell$ unused logical qubits.

Otherwise, we follow the decomposition into addition, lookup and phaseup subroutines and implementation of these subroutines presented in Ref.~\cite{gidney_how_2025}.

\subsubsection{Resource Analysis}
\label{sec:rsa-resources}
\paragraph{Physical Qubits:}
We now determine the number of physical qubits required.
Each working register corresponds to a processing unit with $\kappa(f,\ell,m)$ logical qubits, as given in \cref{eq:kappa}, along with a magic engine.
The number of physical qubits required for the $\rho$ working registers is therefore
\begin{equation}\label{eq:nw}
    n_w=\rho\left(n_{pb}\left\lceil\frac{\kappa(f,\ell,m)}{k}\right\rceil+n_{me}\right).
\end{equation}
Each memory stores $m$ logical qubits and there are $\left\lceil\frac{\rho}{\lceil m/w_1\rceil}\right\rceil$ such memories.
Each memory requires $\nu=\left\lceil m/k\right\rceil$ code blocks of an $\llbracket n,k,d\rrbracket$ code (with $2n$ physical qubits each).
For each processing unit, we also require a port with $n_g+n_b$ additional physical qubits.
The number of physical qubits required for the memory and ports is therefore
\begin{equation}\label{eq:nm}
    n_m=2n\left\lceil\frac{\rho}{\lceil m/w_1\rceil}\right\rceil \left\lceil\frac{m}{k}\right\rceil + \rho (n_g+n_b).
\end{equation}
Hence the total number of physical qubits is $n_\mathrm{total}=n_w+n_m$.

\paragraph{Runtime:}
We define the total number of logical cycles required per iteration of the outer loop of Gidney's algorithm to be $\Sigma$, which is equal to the sum of the final column of \cref{tab:subroutines}.
We define the number of additional logical cycles required to combine the accumulators in each batch of our parallelised version of the algorithm to be $\Lambda$.
When $\rho=1$, $\Lambda=0$; otherwise, we have:
\begin{equation}
\Lambda=27f\lceil \log_2(\rho)\rceil-4f+9\lceil l/w_4 \rceil(2^{w_4}-w_4-1)
\end{equation}
which accounts for the binary tree used to combine accumulators, its uncomputation, Clifford frame cleaning between pairs of accumulators, and the uncomputation of loop 4 on the ancillary working registers.
Since we require $\left\lceil |P|/\rho \right\rceil$ iterations, this implies that the number of logical cycles required for the full outer loop is $\left\lceil |P|/\rho \right\rceil\left(\Sigma+\Lambda\right)$.

Following the completion of this outer loop, we have two further minor steps that are performed once.
First, the uncomputation of loop 1~\cite{gidney_how_2025}---which is performed in parallel on all streams---consists of a single instance of the loop 1 lookup and addition.
Referring to \cref{tab:subroutines}, it takes 
$\upsilon=\lceil m/w_1\rceil\left(6\left(2^{w_1}-w_1+\ell+\mathrm{len}{(m)}-2\right)+2w_1\right)$ logical cycles.
Second, there is a frequency measurement (i.e. inverse quantum Fourier transform followed by measurement); the number of logical cycles required for this is negligible compared to $\Sigma$ so we follow Ref.~\cite{gidney_how_2025} in omitting it from our accounting.
Assuming a perfect success rate of magic engines we can therefore write the total number of logical cycles as
\begin{equation}\label{eq:log-cycles}
    \mathcal{T}'=\left\lceil\frac{|P|}{\rho}\right\rceil\left(\Sigma+\Lambda\right)+\upsilon.
\end{equation}
The proportion of logical cycles requiring $T$ states is approximately (in fact, slightly less than) $2/3$, which leads to an adjusted formula for the true number of logical cycles accounting for the magic engine rejection rate of $p_r$ of 
\begin{equation}
    \mathcal{T}=\left(\frac{2}{3}(1-p_r)^{-1}+\frac{1}{3}\right)\mathcal{T}'.
\end{equation}
Then the total runtime per shot $t=t_l\mathcal{T}$ is the product of the number of logical cycles $\mathcal{T}$ and the time per logical cycle $t_l=d_tt_c$.

The expected number of shots required for the factoring, from Ref.~\cite{gidney_how_2025}, is given by
\begin{equation}\label{eq:shots}
    \sigma = \frac{s+1}{0.99p_S\left(1-2N_\mathrm{RSA}\sqrt{\frac{s+2}{2^{f+1}sw_1}}\right)},
\end{equation}
where $s$ is the Ekera-Håstad parameter, and $p_S$ is the probability that a shot does not have a logical error, which is given by
\begin{equation}
    p_S =\left(1-p_L\right)^{N\mathcal{T}}\left(1-p_T\right)^\tau,
\end{equation}
where $p_L$ is the logical error rate per logical qubit per logical cycle, $p_T$ is the infidelity of output $T$ states from the magic engine, $N$ and $\mathcal{T}$ are the total number of logical qubits and logical cycles per shot and $\tau$ is the $T$ count per shot, which is upper bounded by
\begin{equation}
\tau \leq \frac{2}{3}\left(|P|(\Sigma+\Lambda)+\upsilon\right),
\end{equation}
Hence the expected runtime for the factoring is $t_\mathrm{total}=\sigma t$.

\begin{table*}
    \centering
    \caption{Time cost accounting for all subroutines required for each prime in Gidney's algorithm.
    Loop numbers refer to Gidney's algorithm, as presented in Ref.~\cite{gidney_how_2025}.
    For lookups and phaseups, the size is the address size; for additions, it is the size of the registers being added.
    For addition (loop 2) the size is the average over the loop, since it varies.
    The $T$ Count column is the number of $T$ gates required, which is four times the Toffoli count.
    The Logical Cycles column is the total number of logical cycles, including both $T$ state injections and other logical measurements.
    Since each Toffoli gate requires four $T$ gates and one logical measurement~\cite{jones_lowoverhead_2013}, and one additional logical measurement is required for measurement-based uncompute, this is nearly always $3/2$ times the $T$ count.
    The only exception is Lookup (Loop 1) which also requires additional logical measurements for Clifford frame cleaning.}
    \label{tab:subroutines}
    \resizebox{\textwidth}{!}{
    \begingroup
    \renewcommand{\arraystretch}{1.25}
    \begin{tabular}{|c|c|c|c|c|}
        \hline
        Subroutine & Size & Instances & $T$ Count & Logical Cycles\\
        \hline
        Lookup (Loop 1)
        & $w_1$
        & $\lceil m/w_1\rceil$
        & $4\lceil m/w_1\rceil\left(2^{w_1}-w_1-1\right)$
        & $\lceil m/w_1\rceil\left(6\left(2^{w_1}-w_1-1\right)+2w_1\right)$\\ 
        \hline
        Addition (Loop 1)
        & $\ell+\mathrm{len}{(m)}$
        & $\lceil m/w_1\rceil$
        & $4\lceil m/w_1\rceil\left(\ell+\mathrm{len}{(m)}-1\right)$
        & $6\lceil m/w_1\rceil\left(\ell+\mathrm{len}{(m)}-1\right)$\\
        \hline
        Addition (Loop 2)
        & $\left(2\ell+\mathrm{len}{(m)}+1\right)/2$
        & $4\,\mathrm{len}{(m)}$
        & $8\,\mathrm{len}{(m)}\left(2\ell+\mathrm{len}{(m)}-1\right)$
        & $12\,\mathrm{len}{(m)}\left(2\ell+\mathrm{len}{(m)}-1\right)$\\
        \hline
        Lookup (Loop 3)
        & $2w_3$
        & $4\lceil \ell/w_3\rceil^2 -8\lceil \ell/w_3\rceil +1$
        & $4\left(4\lceil \ell/w_3\rceil^2 -8\lceil \ell/w_3\rceil +1\right)\left(2^{2w_3}-2w_3-1\right)$
        & $6\left(4\lceil \ell/w_3\rceil^2 -8\lceil \ell/w_3\rceil +1\right)(2^{2w_3}-2w_3-1)$\\
        \hline
        Addition (Loop 3)
        & $\ell$
        & $7\lceil \ell/w_3\rceil^2-14\lceil \ell/w_3\rceil$
        & $28\left(\lceil \ell/w_3\rceil^2-2\lceil \ell/w_3\rceil\right)(\ell-1)$
        & $42(\ell-1)\left(\lceil \ell/w_3\rceil^2-2\lceil \ell/w_3\rceil\right)$\\
        \hline
        Lookup (Loop 4)
        & $w_4$
        & $3\lceil \ell/w_4\rceil/2$
        & $6\lceil \ell/w_4\rceil\left(2^{w_4}-w_4-1\right)$
        & $9\lceil \ell/w_4\rceil\left(2^{w_4}-w_4-1 \right)$\\
        \hline
        Addition (Loop 4)
        & $f$
        & $5\lceil \ell/w_4\rceil/2$
        & $10(f-1)\lceil \ell/w_4\rceil$
        & $15(f-1)\lceil \ell/w_4\rceil$\\
        \hline
        Phaseup (Loop 4)
        & $w_4$
        & $\lceil \ell/w_4 \rceil$
        & $4\lceil \ell/w_4 \rceil\left(2^{\lceil w_4/2\rceil}+2^{\lfloor w_4/2\rfloor}-w_4-2\right)$
        & $6\lceil \ell/w_4\rceil \left(2^{\lceil w_4/2\rceil}+2^{\lfloor w_4/2\rfloor}-w_4-2\right)$\\
        \hline
        Phaseup (Loop 3.2)
        & $w_3$
        & $3\lceil \ell/w_3 \rceil^2/2-3\lceil \ell/w_3 \rceil$
        & $6\left(\lceil \ell/w_3 \rceil^2-2\lceil \ell/w_3 \rceil\right)\left(2^{\lceil w_3/2\rceil}+2^{\lfloor w_3/2\rfloor}-w_3-2\right)$
        & $9\left(\lceil \ell/w_3\rceil^2-2\lceil \ell/w_3\rceil\right)
          \left(2^{\lceil w_3/2\rceil}+2^{\lfloor w_3/2\rfloor}-w_3-2\right)$\\
        \hline
        Phaseup (Loop 3.1)
        & $2w_3$
        & $1$
        & $4\left(2^{w_3+1}-2w_3-2\right)$
        & $6\left(2^{w_3+1}-2w_3-2\right)$\\
        \hline
    \end{tabular}
    \endgroup
    }
\end{table*}

\subsubsection{Results}
We now consider the resources---both physical qubits and time---required to factor an RSA-2048 integer on the instantiation of the Pinnacle architecture presented in \cref{sec:instantiation}, given different hardware parameters, namely the code cycle time and physical error rate.
Following Ref.~\cite{gidney_how_2025}, we expect the required logical error rate per logical qubit per logical cycle to be $\lesssim 10^{-14}$ (i.e.,~$\lesssim 10^{-15}$ per code cycle).
With reference to \cref{tab:ler}, this motivates a choice of the $d=24$ GB code architecture for a physical error rate of $p=10^{-3}$ and $d=16$ for the $p=10^{-4}$ architecture.
The precise logical failure rate of the algorithm varies somewhat as other parameters affect the number of logical qubits and logical cycles; this effect is accounted for in the number of shots, given in \cref{eq:shots}.

To determine the minimal required physical qubits and runtime, we optimise over the algorithmic parameters of Gidney's algorithm with the following ranges:
\begin{itemize}
    \item Ekerå-Håstad parameter, $1\leq s \leq 16$;
    \item Accumulator truncation, $24 \leq f \leq 59$;
    \item Prime bit length, $18 \leq \ell \leq 25$;
    \item Loop 3 window size, $2\leq w_3 \leq 6$;
    \item Loop 4 window size, $2\leq w_4 \leq 6$.
\end{itemize}
We also optimise the parallelisation factor over the range $1\leq\rho\leq |P|$, where $|P|\approx nm/\ell w_1$ is the number of primes in the residue system.
Recall that, unlike in Ref.~\cite{gidney_how_2025}, we fix the loop 1 window size as $w_1=k/2$.
Following Ref.~\cite{gidney_how_2025}, we also impose the feasibility condition that the number of primes of bit length $\ell$, $\pi(\ell)\approx 2^{\ell-1}/\ell\ln(2)$, cannot be smaller than the number of primes $|P|$.
The results of this optimisation are shown in \cref{tab:results} and \cref{fig:heatmaps}.

\begin{table}
    \centering
    \caption{Minimum number of physical qubits required to complete factoring in a range of expected runtimes for a range of hardware parameters.
    Mq and kq represent megaqubits (i.e., $\times 10^6$ qubits) and kiloqubits ($\times 10^3$ qubits) respectively.}
    \label{tab:results}
    \begingroup
    \renewcommand{\arraystretch}{1.2}
    \begin{tabular}{|c|c|c|c|c|c|}
        \hline
        Code & Physical & \multicolumn{4}{c|}{Physical Qubits for Runtime $\leq$} \\
        Cycle & Error Rate & 1 year & 1 month & 1 week & 1 day \\
        \hline
        \multirow{2}{*}{1 \textmu s} & $10^{-3}$ & 94 kq & 94 kq & 135 kq & 381 kq \\ 
            & $10^{-4}$ & 53 kq & 53 kq & 64 kq & 141 kq \\ 
        \hline
        \multirow{2}{*}{10 \textmu s} & $10^{-3}$ & 94 kq & 193 kq & 501 kq & 2.9 Mq \\ 
            & $10^{-4}$ & 53 kq & 82 kq & 176 kq & 845 kq \\ 
        \hline
        \multirow{2}{*}{100 \textmu s} & $10^{-3}$ & 179 kq & 1.0 Mq & 4.1 Mq & 30 Mq \\ 
            & $10^{-4}$ & 77 kq & 323 kq & 1.2 Mq & 8.2 Mq \\ 
        \hline
        \multirow{2}{*}{1 ms} & $10^{-3}$ & 871 kq & 9.5 Mq & 44 Mq & - \\ 
            & $10^{-4}$ & 278 kq & 2.7 Mq & 12 Mq & 108 Mq \\ 
        \hline
    \end{tabular}
    \endgroup
\end{table}

We find that fewer than one hundred thousand physical qubits are required for factoring at a physical error rate of $p=10^{-3}$ in an expected runtime of one month.
Alternatively, with the same error rate and code cycle time, factoring can be completed in one week with 139 thousand physical qubits---compared to one million in Ref.~\cite{gidney_how_2025}---or in one day with 400 thousand physical qubits.
With one million physical qubits, factoring takes an expected time of eight hours, compared with five days in Ref.~\cite{gidney_how_2025}.

We can also consider our results in regimes relevant to other hardware platforms.
For example, at a physical error rate of $p=10^{-4}$---relevant to trapped ions~\cite{hughes_trappedion_2025}---the minimum number of physical qubits required for factoring is 53 thousand.
For a typical trapped ion code cycle time of $1$ ms~\cite{litinski_how_2023}, factoring with one million physical qubits takes an expected time of less than three months.
Factoring in a shorter time with these parameters is also possible with a feasible number of physical qubits.
For example, factoring can be completed in one month with 2.7 million physical qubits or in one week with 12 million physical qubits.
By comparison, adjusting the results of Ref.~\cite{gidney_how_2025} to such a code cycle time would give a prohibitively long runtime, while Ref.~\cite{beverland_assessing_2022} estimated a runtime of 3 years with 8.6 million qubits for a surface code architecture with trapped ion parameters.

Our results are also comparable to the results of Ref.~\cite{zhou_resource_2025}, which apply to a neutral atom platform with physical error rate of $p=10^{-3}$ and a code cycle time of 1 ms. 
Specifically, with 19 million qubits we find a runtime of 15 days, compared with the result of 5.6 days in Ref.~\cite{zhou_resource_2025}
The low runtime achieved by that work is achieved by using transversal gates with algorithmic fault tolerance to reduce the logical cycle time to be equal to the code cycle time, compared with $d_t=26$ code cycles on our architecture.
Our architecture has the potential to support an analogous reduction by incorporating fast surgery~\cite{baspin_fast_2025}, with the potential to achieve an order-of-magnitude lower runtime or physical qubit number.
We look forward to realising this potential in future work.

\section{Conclusion}
We have presented the Pinnacle Architecture, which leverages the high encoding rate of QLDPC codes to achieve universal quantum computing with order-of-magnitude overhead reductions compared with surface code architectures.
In particular, we have shown that factoring 2048-bit RSA integers, which requires close to a million physical qubits using surface code architectures~\cite{gidney_how_2025}, can be done with fewer than one hundred thousand physical qubits on the Pinnacle Architecture.
Given the challenges posed in scaling from one hundred thousand to one million physical qubits, such as the need on many hardware platforms for networking between separated devices~\cite{mohseni_how_2025}, this has the potential to significantly hasten the onset of practical quantum computing.

Importantly, this is only the beginning of the story for QLDPC architectures.
While it has been suggested that reducing the physical qubit count for 2048-bit RSA factoring by an order of magnitude (i.e., to one hundred thousand physical qubits) is implausible on surface code architectures~\cite{gidney_how_2025}, the same cannot be said of QLDPC architectures.
Indeed, a range of higher rate QLDPC codes are known than the generalised bicycle codes used here~\cite{breuckmann_quantum_2021}.
Incorporating such codes into the Pinnacle Architecture, combined with further optimisation of its components, could plausibly achieve such a reduction.
As such, this work serves not only as a major step forward in its own right, but also a foundation on which we expect to make substantial further progress.

\section{Acknowledgements}
We thank Stephen Bartlett and Kevin Obenland for discussions, Calida Tang for organisational support, Fernando Borretti for software engineering support, Max McIsted for assistance with preparing the figures, and Scott Aaronson for thoughtful feedback on the title.
We also thank Marcel Hinsche for identifying a technical error in the method used to combine accumulators during factoring in the first version of this paper, and for assisting in its correction.

\appendix

\section{Cost of Clifford Frame Cleaning}
\label{sec:cleaning}
In this appendix, we show that Clifford frame cleaning can be completed using the number of logical measurements claimed in \cref{sec:parallel-operation}.
Specifically, we show that it can be completed in the same number of Pauli $\pi/4$ rotations, which we write as $R_{\pi/4}(P)$ for a Pauli rotation axis $P$.
Each such rotation can be implemented by a joint logical Pauli measurement with a single logical ancilla in the $\ket{\bar{0}}$ state~\cite{litinski_game_2019}.

We present our proofs using the Pauli and Clifford operator representation of Ref.~\cite{aaronson_improved_2004}.
Specifically, an $n$-qubit Pauli operator is associated (uniquely, up to a global phase) with an element $v\in\mathbb{Z}_2^{2n}$ by the expression $P_v=\prod_{i=1}^n X_i^{v_i}Z_i^{v_{n+i}}$.
In this formalism, $P_uP_v=P_{u+v}$.
Commutation relations are specified by the symplectic inner product
\begin{equation}
    \langle u,v\rangle = uJv^T,\quad J=\begin{bmatrix}0 & I_n \\ I_n & 0\end{bmatrix},
\end{equation}
and $u$ and $v$ are treated as row vectors.
Specifically, $P_u$ and $P_v$ commute if $\langle u,v\rangle=0$, and anti-commute if $\langle u,v\rangle=1$.
The symplectic complement $W^\perp$ of subspace $W$ is a subspace whose elements commute with all elements in $W$.

In the same formalism, an $n$-qubit Clifford operator $U$ can be represented by a $2n\times 2n$ matrix $M_U$ that preserves the symplectic inner product (i.e., such that $M_UJM_U^T=J$).
The action of $U$ on $P_v$ by conjugation corresponds to matrix multiplication on the right (i.e., $P\mapsto UPU^\dag$ corresponds to $v\mapsto vM_U$).
The product of two Clifford operators $VU$ therefore corresponds to the product of their matrix representations $M_UM_V$.
A Pauli $\pi/4$ rotation $R_{\pi/4}(Q)$ is a special type of Clifford operator which is specified by an $n$-qubit Pauli operator $Q$.
It acts as $R_{\pi/4}(Q)PR_{\pi/4}(Q)^\dag=PQ$ if $P$ and $Q$ anti-commute, and as $R_{\pi/4}(Q)PR_{\pi/4}(Q)^\dag=P$ if $P$ and $Q$ commute.
Therefore, $R_{\pi/4}(P_u)$ acts on $v\in\mathbb{Z}_2^{2n}$ as $E_u(v)=v+\langle u,v\rangle u$.

There are two cases of Clifford frame cleaning to consider: the general case, where $4w$ steps are required to clean off $w$ qubits; and the case of cleaning a memory port, where only $2w$ steps are required.
We first consider the general case.

\begin{lemma}\label{lem:cleaning1}
Let $U$ be an $n$-qubit Clifford operator.
Then for $w\leq n$, there exists a sequence of $4w$ Pauli operators $P_1,\dots,P_{4w}$ such that $R_{\pi/4}(P_1)R_{\pi/4}(P_2)\ldots R_{\pi/4}(P_{4w})U$ is a Clifford operator supported only on the last $n-w$ qubits.
\end{lemma}
\begin{proof}
We proceed by induction.
Let
\begin{equation}
    M^{(k)}=
    \left[
    \begin{array}{cc|cc}
        I_k & 0  &  0   & 0 \\
        0   & *  &  0   & * \\
        \hline
        0   & 0  &  I_k & 0 \\
        0   & *  &  0   & *
    \end{array}
    \right].
\end{equation}
We will show for $1\le k\le w$, given a matrix of the form $M^{(k-1)}$, that there exists a product of four Pauli $\pi/4$ rotations with matrices $E_{\alpha_1}$, $E_{\alpha_2}$, $E_{\alpha_3}$, and $E_{\alpha_4}$, such that $E_{\alpha_4}E_{\alpha_3}E_{\alpha_2}E_{\alpha_1}M^{(k-1)}$ is of the form $M^{(k)}$.
It follows that $U=M^{(0)}$ can be mapped to an operator $M^{(w)}$ that has support only on the last $n-w$ qubits with $4w$ Pauli $\pi/4$ rotations.

Let $v^{(k)}$ denote the $k$th row of $M^{(k-1)}$, $e_k$ denote the $k$th vector in the standard symplectic basis and $f_k=e_{n+k}$, and let $A_{k-1}$ be the symplectic subspace spanned by $\{e_1,\ldots e_{k-1},f_1,\ldots f_{k-1}\}$ corresponding to the Pauli group on the first $k-1$ qubits.
We will choose $\alpha_1,\alpha_2,\alpha_3,\alpha_4\in A_{k-1}^\perp$ to ensure that these operations act as the identity on all $a\in A_{k-1}$.
We now consider two cases, and in each choose $\alpha_1,\alpha_2$ such that they map $v^{(k)}$ to $e_k$.

First, if $\langle v^{(k)},e_k\rangle=1$, let $\alpha_1=v^{(k)}+e_k$ and $\alpha_2=0$.
Then $\langle v^{(k)},v^{(k)}+e_k\rangle=1$, and so
\begin{equation}
    E_{v^{(k)}+e_k}(v^{(k)})=e_k.
\end{equation}

Second, if $\langle v^{(k)},e_k\rangle=0$, then if $\langle v^{(k)},f_k\rangle=1$ let $u_k=f_k$, and if $\langle v^{(k)},f_k\rangle=0$ let $u_k=f_k+e_{n+\delta\pmod{2n}}$, where $\delta>k$ is the position of the first nonzero element of $v^{(k)}$.
Moreover, let $\alpha_1=v^{(k)}+u_k$ and $\alpha_2=e_k+u_k$.
Then $\langle v^{(k)},v^{(k)}+u_k\rangle=\langle u_k,e_k+u_k\rangle=1$, and so
\begin{equation}
    E_{e_k+u_k}E_{v^{(k)}+u_k}\big(v^{(k)}\big)=E_{e_k+u_k}(u_k)=e_k.
\end{equation}

Let $\tilde{M}^{(k-1)}=E_{e_k+u}E_{v^{(k)}+u}\big(M^{(k-1)}\big)$ and the $i$th row of $\tilde{M}^{(k-1)}$ be $\tilde{v}^{(i)}$.
As this preserves the symplectic product, $1=\langle v^{(k)},v^{(n+k)}\rangle=\langle\tilde{v}^{(k)},\tilde{v}^{(n+k)}\rangle=\langle e_k,\tilde{v}^{(n+k)}\rangle$.
We now consider two cases, and in each choose $\alpha_3,\alpha_4$ such that they fix $e_k$ and map $\tilde{v}^{(n+k)}$ to $f_k$.

First, if $\langle\tilde{v}^{(n+k)},f_k\rangle=1$, let $\alpha_3=\tilde{v}^{(n+k)}+f_k$ and $\alpha_4=0$.
Then $\langle\tilde{v}^{(n+k)},\tilde{v}^{(n+k)}+f_k\rangle=1$, and so
\begin{equation}
    E_{\tilde{v}^{(n+k)}+f_k}(\tilde{v}^{(n+k)})=f_k.
\end{equation}
We also have $\langle e_k,\tilde{v}^{(n+k)}+f_k\rangle=0$, which implies that $E_{\tilde{v}^{(n+k)}+f_k}(e_k)=e_k$.

Second, if $\langle\tilde{v}^{(n+k)},f_k\rangle=0$, let $\alpha_3=\tilde{v}^{(n+k)}+e_k+f_k$ and $\alpha_4=e_k$.
Then $\langle\tilde{v}^{(n+k)},\tilde{v}^{(n+k)}+e_k+f_k\rangle=1$, and so
\begin{equation}
    E_{e_k}E_{\tilde{v}^{(n+k)}+e_k+f_k}(\tilde{v}^{(n+k)})=E_{e_k}(e_k+f_k)=f_k.
\end{equation}
We also have $\langle e_k, \tilde{v}^{(n+k)}+e_k+f_k\rangle=0$, which implies that $E_{e_k}E_{\tilde{v}^{(n+k)}+e_k+f_k}(e_k)=e_k$.

Hence in all cases
\begin{equation}
    M^{(k)}=E_{\alpha_4}E_{\alpha_3}E_{\alpha_2}E_{\alpha_1}M^{(k-1)},
\end{equation}
which completes the proof.
\end{proof}

Now we consider the specific case of cleaning a memory port.

\begin{lemma}\label{lem:cleaning2}
Let $U$ be an $n$-qubit Clifford operator that is a product of Clifford operators that act either trivially or as the control of a CNOT on the first $w\le n$ qubits.
Then there exists a sequence of $2w$ Pauli operators $P_1,\dots,P_{2w}$ such that $R_{\pi/4}(P_1)R_{\pi/4}(P_2)\ldots R_{\pi/4}(P_{2w})U$ is a Clifford operator supported only on the last $n-w$ qubits.
\end{lemma}
\begin{proof}
Since CNOTs commute with $Z$ operators on their control qubits, $U$ must commute with all $Z$-type operators with support on the first $w$ qubits.
Consequently, the matrix representation of $U$ has the form
\begin{equation}
    U=
    \left[
    \begin{array}{cc|cc}
        I_{w} & *  & *   & * \\
        0     & *  & *   & * \\
        \hline
        0     & 0  & I_w & 0 \\
        0     & *  & *   & * 
    \end{array}
    \right].
\end{equation}

We proceed by induction, analogously to the proof of \cref{lem:cleaning1}.
Let
\begin{equation}
    M^{(k)}=
    \left[
    \begin{array}{ccc|ccc}
        I_k & 0       & 0  &  0   & 0       & 0 \\
        0   & I_{w-k} & *  &  0   & *       & * \\
        0   & 0       & *  &  0   & *       & * \\
        \hline
        0   & 0       & 0  &  I_k & 0       & 0 \\
        0   & 0       & 0  &  0   & I_{w-k} & 0 \\
        0   & 0       & *  &  0   & *       & * 
    \end{array}
    \right].
\end{equation}
We will show for $1\le k\le w$, given a matrix of the form $M^{(k-1)}$, that there exists a product of two Pauli $\pi/4$ rotations with matrices $E_{\alpha_1}$ and $E_{\alpha_2}$, such that $E_{\alpha_2}E_{\alpha_1}M^{(k-1)}$ is of the form $M^{(k)}$.
It follows that $U=M^{(0)}$ can be mapped to an operator $M^{(w)}$ that has support only on the last $n-w$ qubits with $2w$ Pauli $\pi/4$ rotations.

Let $v^{(k)}$ denote the $k$th row of $M^{(k-1)}$, $e_k$ denote the $k$th vector in the standard symplectic basis and $f_k=e_{n+k}$, and let $A_{k-1}$ be the symplectic subspace spanned by $\{e_1,\ldots e_{k-1},f_1,\ldots f_{k-1}\}$ corresponding to the Pauli group on the first $k-1$ qubits.
We will choose $\alpha_1,\alpha_2\in A_{k-1}^\perp$ to ensure that these operations act as the identity on all $a\in A_{k-1}$, and such that they also fix $f_j$ for $k\le j\le w$ and map $v^{(k)}$ to $e_k$.
We now consider two cases, noting that the structure of $M^{(k-1)}$ implies that $\langle v^{(k)},f_j\rangle=\delta_{jk}$ for $1\le j\le w$.

First, if $\langle v^{(k)},e_k\rangle=1$, let $\alpha_1=v^{(k)}+e_k$ and $\alpha_2=0$.
Then $\langle v^{(k)},v^{(k)}+e_k\rangle=1$, and so
\begin{equation}
    E_{v^{(k)}+e_k}\big(v^{(k)}\big)=e_k.
\end{equation}
We also have $\langle f_j,v^{(k)}+e_k\rangle=2\delta_{jk}=0$ for $k\le j\le w$, which implies that $E_{v^{(k)}+e_k}(f_j)=f_j$.

Second, if $\langle v^{(k)},e_k\rangle=0$, let $\alpha_1=v^{(k)}+e_k+f_k$ and $\alpha_2=f_k$.
Then $\langle v^{(k)},v^{(k)}+e_k+f_k\rangle=1$, and so
\begin{equation}
    E_{f_k}E_{v^{(k)}+e_k+f_k}\big(v^{(k)}\big)=E_{f_k}(e_k+f_k)=e_k.
\end{equation}
We also have $\langle f_j,v^{(k)}+e_k\rangle=2\delta_{jk}=0$ for $k\le j\le w$, which implies that $E_{f_k}E_{v^{(k)}+e_k+f_k}(f_j)=f_j$.

Hence in all cases
\begin{equation}
    M^{(k)}=E_{\alpha_2}E_{\alpha_1}M^{(k-1)},
\end{equation}
which completes the proof.
\end{proof}

\bibliography{refs}

\end{document}